\documentclass{article}
\usepackage{fullpage}
\usepackage[dvipsnames]{xcolor}
\usepackage{hyperref}
\hypersetup{
    colorlinks=true,
    linkcolor=blue,
    citecolor=violet,
    filecolor=magenta,      
    urlcolor=cyan,
}

\usepackage{authblk}
\usepackage{graphicx,graphbox}
\graphicspath{{./figures/}}
\usepackage{wrapfig}

\usepackage{tikz-cd}

\usepackage{amssymb,amsmath,amsthm,bbm}
\usepackage{mathtools}%
\usepackage[capitalize]{cleveref}

\usepackage{amsthm,amssymb,amsmath}

\usepackage{thm-restate}
\usepackage{enumitem}

\numberwithin{equation}{section}

\newtheorem{theorem}[equation]{Theorem}
\newtheorem{cor}[equation]{Corollary}
\newtheorem{lemma}[equation]{Lemma}

\newtheorem{proposition}[equation]{Proposition}

\newtheorem{definition}[equation]{Definition}

\Crefname{defn}{Defn.}{Defns.}
\Crefname{lemma}{Lem.}{Lems.}
\Crefname{alg}{Alg.}{Algs.}

\usepackage[backend = biber, url=false,sorting=nyt,style=numeric-comp,maxbibnames = 99]{biblatex}
\usepackage[colorinlistoftodos,prependcaption,textsize=small]{todonotes}

\definecolor{ltgray}{gray}{0.9}
\definecolor{gray}{rgb}{0.95,0.95,0.96}
\definecolor{dkgray}{rgb}{0.7,0.7, 0.735}
\definecolor{ltblue}{rgb}{0.55,0.55, 0.95}
\definecolor{ltGreen}{rgb}{0.25,0.65, 0.25}
\definecolor{dkgreen}{RGB}{0, 100, 0}
\definecolor{dkred}{rgb}{0.75,0.0, 0.0}
\definecolor{dkorange}{rgb}{.82,.45, 0.0}
\definecolor{ltred}{rgb}{0.95,0.95, 0.85}
\definecolor{utahRed}{rgb}{.8, 0, 0}
\definecolor{oregonGreen}{rgb}{0, .41, .163}
\definecolor{albanyPurple}{rgb}{0.4, 0, .55}

\newcommand{\R}{\mathbb{R}}

\newcommand{\X}{{\mathbb X}}
\newcommand{\Y}{{\mathbb Y}}
\newcommand{\Z}{{\mathbb Z}}

\newcommand{\cC}{\mathcal{C}}
\newcommand{\cD}{\mathcal{D}}

\newcommand{\cU}{\mathcal{U}}

\newcommand{\Open}{\mathbf{Open}}
\newcommand{\Set}{\mathbf{Set}}
\newcommand{\Int}{\mathbf{Int}}

\newcommand*{\Parallelogramr}[1][]{
  \pgfpicture\pgfsetroundjoin
    \pgftransformxslant{.6}
    \pgfpathrectangle{\pgfpointorigin}{\pgfpoint{.60em}{.65em}}
    \pgfusepath{stroke,#1}
  \endpgfpicture}
  
\newcommand*{\Parallelograml}[1][]{
  \pgfpicture\pgfsetroundjoin
    \pgftransformxslant{-.6}
    \pgfpathrectangle{\pgfpointorigin}{\pgfpoint{.60em}{.65em}}
    \pgfusepath{stroke,#1}
  \endpgfpicture}
  
\newcommand{\triangled}{\raisebox{\depth}{$\bigtriangledown$}}
\newcommand{\triangleu}{\bigtriangleup}

\usepackage{scalerel}
\newcommand{\Lpl}{L_{\scaleobj{.7}{\Parallelograml}}}
\newcommand{\Lpr}{L_{\scaleobj{.7}{\Parallelogramr}}}
\newcommand{\Ltd}{L_{\bigtriangledown}}
\newcommand{\Ltu}{L_{\bigtriangleup}}

\newcommand{\e}{\varepsilon}
\renewcommand{\phi}{\varphi}
\newcommand{\inv}{^{-1}}

\newcommand{\1}{\mathbbm{1}}

\newcommand{\para}[1]        {\vspace{0mm}\noindent{\textbf{#1}}}

\newcommand{\tF}{\widetilde{F}}
\newcommand{\tG}{\widetilde{G}}
\newcommand{\tphi}{\widetilde{\phi}} 
\newcommand{\tpsi}{\widetilde{\psi}} 

\newcommand{\phitt}{\texttt{phi}}
\newcommand{\psitt}{\texttt{psi}}

\newcommand {\mm}[1] {\ifmmode{#1}\else{\mbox{\(#1\)}}\fi}

\newcommand{\Top}{\mathbf{Top}}

\begin{document}
\title{Bounding the Interleaving Distance for Mapper Graphs with a Loss Function
 \thanks{The authors would like to thank Vidit Nanda for helpful discussions related to non-commutative diagrams, as well as anonymous reviewers whose feedback improved the state of the paper. This work was funded in part by the National Science Foundation through grants CCF-1907591, CCF-1907612, CCF-2106578, CCF-2142713, CCF-2106672, DMS-2301361, and IIS-2145499.}
}

\author[1]{Erin Chambers}
\author[2]{Elizabeth Munch}
\author[3]{Sarah Percival}
\author[4]{Bei Wang}
\affil[1]{University of Notre Dame}
\affil[2]{Michigan State University}
\affil[2]{University of New Mexico}
\affil[4]{University of Utah}

\date{}

\maketitle

\begin{abstract}
Data consisting of a graph with a function mapping into $\R^d$ arise in many data applications, encompassing structures such as Reeb graphs, geometric graphs, and knot embeddings. 
As such, the ability to compare and cluster such objects is required in a data analysis pipeline, leading to a need for distances between them.  
In this work, we study the interleaving distance on discretization of these objects, called mapper graphs when $d=1$, where functor representations of the data can be compared by finding pairs of natural transformations between them. 
However, in many cases, computation of the interleaving distance is NP-hard.
For this reason, we take inspiration from recent work by Robinson to find quality measures for families of maps that do not rise to the level of a natural transformation, called assignments. 
We then endow the functor images with the extra structure of a metric space and define a loss function which measures how far an assignment is from making the required diagrams of an interleaving commute. 
Finally we show that the computation of the loss function is polynomial with a given assignment.
We believe this idea is both powerful and translatable, with the potential to provide approximations and bounds on interleavings in a broad array of contexts.

\end{abstract}

\section{Introduction}
\label{sec:introduction}

Graphs with additional geometric information arise in many contexts in data analysis. 
For instance, a \emph{geometric graph} is generally defined as an abstract graph along with a well-behaved embedding of the graph into $\R^2$, while a graph with a well-behaved map into $\R$ is called a \emph{Reeb graph}.
In particular from the viewpoint of the Reeb graph, these types of input data can arise by studying connected component structures from more general input $\R^d$-spaces, meaning a topological space $\X$ with a function $f:\X \to \R^d$. 
Such graphs are a fundamental object used to model a wide range of data sets, ranging from maps and trajectories \cite{Ahmed2015, Buchin2023}, to commodity networks (e.g.~transportation networks \cite{Jiang2022,Myers2023}) and shape skeletons for object recognition \cite{Ayzenberg2022, Zeng2021}. 
The ability to compare, cluster, and simplify such representative objects is essential in a data analysis pipeline, leading to a need for theoretically motivated and computable distances. 
In this paper, we study a distance for a discretization of the input data, known as mapper \cite{Singh2007}. 
That is, starting from a topological space $\X$ with a function $f:\X \to \R^d$ (resp.~a point cloud $P$ with a function $f: P \to \R^d$), mapper is an encoding of the connected components (resp.~clusters) of $f\inv(U_\alpha)$ for some cover $\cU = \{U_\alpha\}$ of $\R^d$ defined as the nerve of the pullback cover. 
When $d=1$, this results in a graph structure called a \emph{mapper graph} (see Fig.~\ref{fig:geomgraphs}).
Since there are higher dimensional cells, we call the resulting construction for $d>1$ an $\R^d$-mapper complex; however we abuse terminology and generally call this construction a mapper graph whether the dimension is 1 or not.

There has been extensive work on metrics for general graphs, geometric graphs, and Reeb graphs 
(see surveys \cite{Deza2013,Conci2017,Donnat2018,Wills2020}, \cite{Buchin2023}, and \cite{YanMasoodSridharamurthy2021, Bollen2021} respectively). 
In this paper, we will draw inspiration from the interleaving distance;
specifically, we develop a natural extension of the interleaving distance on Reeb graphs~\cite{deSilva2018} to the setting of mapper graphs.   
Interleaving distances arose in the context of generalizing the bottleneck distance for persistence modules \cite{Chazal2009} and were subsequently translated to more general categorical frameworks in \cite{Bubenik2014a, deSilva2018}. 
With the exception of 1-parameter persistence \cite{Lesnick2015}, the interleaving distance is NP-hard in many contexts including multi-parameter persistence \cite{Bjerkevik2018,Bjerkevik2019}, and Reeb graphs~\cite{deSilva2018}.
However, some additional structural information can give better algorithms such as FPT computation for merge trees \cite{FarahbakhshTouli2019}, and  polynomial time for formigrams \cite{Kim2019b} and labeled merge trees \cite{Munch2018,Gasparovic2019}. 
Indeed, the closest work to our approach is work providing bounds for the interleaving distance restricted to merge trees: 
\cite{Curry2022} use the Gromov-Wasserstein distance to find a leaf labeling that gives an upper bound using the easy to compute labeled interleaving distance \cite{Gasparovic2019, Munch2018}; 
while \cite{Pegoraro2021} uses the map formulation of \cite{Morozov2013,FarahbakhshTouli2019} with an integer linear program to provide a bound.
See~\cite{Bjerkevik2018} for a recent summary of interleaving distance complexity results.

When $d=1$, there is already work using the interleaving distance to relate the Reeb graph and its mapper graph \cite{Carriere2017,Carriere2018,Munch2016,Brown2020,botnan2020}. 
We will encode the structure of our more general $\R^d$-mapper complexes in a discretized setting by imposing a grid structure $K$ on $\R^d$. 
Then we can represent the input data $f:X \to \R^d$ as a cosheaf of the form $F:\Open(K) \to \Set$ where we store the path-connected components of inverse images of open sets $\pi_0(f\inv(U))$. 
The idea of the interleaving distance, in this context, is to compare two cosheaves $F,G:\Open(K) \to \Set$ using a pair of natural transformations $\phi:F \Rightarrow G^n$ and $\psi:G \Rightarrow F^n$ mapping into relaxations of the original inputs.   
The complexity of computing this distance then relies on finding the smallest $n$ with available $\phi$ and $\psi$ maps, which in our setting immediately connects to hard underlying problems such as graph isomorphism.

The ideas building this distance are rooted in previous work that study interleavings in related contexts. 
In some sense, we can view this distance as a  discretized cosheaf version of the continuous sheaf version of the interleaving distance that was previously studied~\cite{Curry2014}. 
It can fit into the more general framework of an interleaving distance arising from a category with a flow \cite{deSilva2016}, or as an interleaving distance on generalized persistence modules with a family of translations \cite{Bubenik2014a}, but prior work in these areas focused on theoretical properties and did not address computational aspects, as the more general framework makes such study incredibly difficult. 
Perhaps the closest version of this distance is mentioned as a special case of a general categorical framework~\cite{botnan2020}; however, that setting keeps the thickening of the open sets structure tightly bound to thickening in $\R^d$, whereas we choose to define the distance entirely over the combinatorial structures.

\begin{figure}
    \centering
    \includegraphics[width = .4\textwidth]{example2.pdf}
    \qquad
    \qquad
    \includegraphics[width = .4\textwidth]{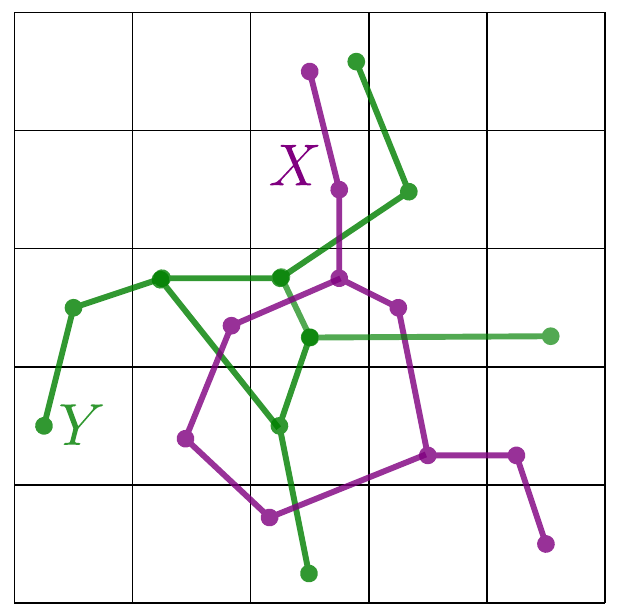}
    \caption{(Left) An example of a Reeb graph ($d=1$) where the discretization of $\R$ is used to generate a mapper graph. (Right) An example of a pair of geometric graphs ($d=2$) which we study using the connected components restricted to the grid. }
    \label{fig:geomgraphs}
\end{figure}

While all this prior work is powerful in theory, the computational complexity of the construction in more general settings has meant a lack of the use of the interleaving distance in practice.
To circumvent these issues, we take inspiration from recent work of Robinson \cite{Robinson2020} to define quality measures for families of maps that do not rise to the level of a natural transformation, in order to allow for non-optimal maps $\phi$ and $\psi$ in this framework.  
We then apply these quality measures to $\R^d$-mapper complexes, in the hopes of utilizing algorithms from geometry and graph theory to make computation more feasible.

In particular, in \cite{Robinson2020} the object of study is a single input assignment of data of the form $P: \Open(X) \to \Set$ and, with the added structure of a pseudometric for each set $P(U)$, provides a measurement for how far the input data is from having a global section. 
In our work, we instead work with a pair of functors $F,G:\Open(K) \to \Set$ as input, and study collections of maps 
${\phi = \{\phi_U: F(U) \to G(U^n)\mid U\}}$ and $\psi = \{\psi_U:G(U) \to F(U^n)\mid U\}$ which we call an \textit{assignment} when they do not necessarily form a true interleaving.  
We then endow the image with the extra structure of a metric space, so that we have pairs $(F(U),d_U)$ for every open set $U$. 
Using this metric structure, we define a loss function $L(\phi,\psi)$ which measures how far the required diagrams of an interleaving are from commuting, given any input assignment (Thm.~\ref{thm:bound}).  
We modify this bound by only focusing on the loss function computed for a basis of the topology, $L_B(\phi,\psi)$ (Thm.~\ref{thm:secondBound}), which not only improves the computational complexity but also improves the bound. 
Then, we show that the computation of the bound is polynomial, opening up the potential for algorithmic approximation of the interleaving distance. 
Throughout, we show examples encoding the data of a \textit{geometric graph} (i.e.~a graph $G$ with a straight line embedding $f:G \to \R^2$) or a Reeb graph (a graph $G$ with a straight line map to $\R$); see Fig.~\ref{fig:geomgraphs} for an illustration.  However, this kind of input is not a requirement for our framework.
 
We note that this paper is the first step in a larger project. 
That is, the paper here presents a loss function that can be computed given an input $n$-assignment $\phi, \psi$ and results in a bound on the distance by explicitly constructing an $\left(n+L_B(\phi,\psi)\right)$-interleaving. 
As with many garbage-in-garbage-out settings, this bound is only as good as the input, but this study seeks to determine the most general possible bound with no guarantees on the input at all. 
In the followup work \cite{Chambers2025},  we include this loss function as part of an optimization strategy to update a given assignment in order to find better bounds as well as provide further study on how close to optimal is possible. 

\para{Outline.}
In Sec.~\ref{sec:background} we provide the necessary technical background to set up the interleaving distance for  mapper graph inputs. 
In Sec.~\ref{sec:loss-function}, we define the loss functions and bounds. 
We discuss algorithmic requirements of the bound in Sec.~\ref{sec:computation}.
Next, we show how this loss function can be used to similarly bound the Reeb graph interleaving distance by approximating the Reeb graph with a mapper graph in Sec.~\ref{sec:ReebLoss}.
We include all technical proofs in Sec.~\ref{sec:technicalProofs}.
Finally, we discuss broader implications of this work in Sec.~\ref{sec:discussion}. 

\section{Technical Background}
\label{sec:background}

We will assume several example types of inputs in this paper. 
All are tied together by having data of the form $f:\X \to \R^d$, where $\X$ is a topological space.
In particular, we will require that $\pi_0(f\inv(U))$ is a finite set for some reasonable collection of open sets $U \subset \R^d$. 
We thus assume that $\X \subset \mathbb{R}^k$ is a semi-algebraic set and that $f$ is a semi-algebraic map since this results in our desired restrictions.

\subsection{Functors and Cosheaves}
\label{ssec:cosheaf}

We give basic definitions for the category theoretic notions required in this paper, and direct the interested reader to \cite{Riehl2017,Curry2014} for further details. 
A \emph{category} $\cC$ consists of a collection of objects $X,Y,Z,\dots$ and morphisms $f,g,h,\dots$ with the following requirement: morphisms $f:X \to Y$ have designated domain $X$ and codomain $Y$; every object has a designated identity morphism $\1_X: X \to X$, and any pair of morphisms $f:X \to Y$ and $g:Y \to Z$ has a composite morphism $g \circ  f:X \to Z$. 
These objects and morphisms satisfy an identity axiom, where $f:X \to Y$ is the same as the $\1_Y \circ  f$ and $f \circ \1_X$; and composition (denoted by $\circ$ but often dropped when unnecessary) is associative, so $h\circ (g\circ f) = (h\circ g)\circ f$. 
Some example categories are $\Set$ where objects are finite sets and morphisms are set maps; 
$\Top$ where objects are topological spaces and morphisms are continuous functions;
and $\Open(\X)$ for a given topological space $\X$, where the objects are open sets and morphisms are given by inclusion. 

A \emph{functor} $F:\cC \to \cD$ is a map between categories preserving the relevant structures.
Specifically, for every object $X \in \cC$ there is a an object $F(X) \in \cD$, and for every morphism $f:X \to Y$, there is a morphism $F[f]:F(X) \to F(Y)$. 
To be a functor, $F$ must further satisfy that for  any $X \in \cC$, $F[\1_X] = \1_{F(X)}$ and for any composable pair $f,g \in \cC$, we have $F[gf] = F[g] F[f]$. 
Given functors $F,G: \cC \to \cD$, a \emph{natural transformation} $\eta: F \Rightarrow G$  consists of a map $\eta_X:F(X) \to G(X)$ for every $X \in \cC$ (called the components) so that for any morphism $f:X \to Y$ in $\cC$, the following diagram 
\begin{equation*}
\begin{tikzcd}
X\ar[d,  "f"] 
& F(X) 
    \ar[r, "\eta_X"] 
    \ar[d, "{F[f]}"']
& G(X) 
    \ar[d, "{G[f]}"]
\\
Y 
& F(Y) \ar[r, "\eta_Y"] 
& G(Y)
\end{tikzcd}
\end{equation*}
commutes. 
One example is $\pi_0:\Top \to \Set$, where $\pi_0(\X)$ is the set of path-connected components of $\X$, and morphisms are set maps $\pi_0[f]: \pi_0(\X) \to \pi_0(\Y)$ sending a connected component $A$ in $\X$ to the connected component of $f(A)$ in $\Y$.
Note that throughout the paper, we use the term \textit{component} to mean \textit{path-connected component}. 

A diagram is a functor $F:J \to \cC$ where $J$ is a small category. 
In essence, this construction picks out a collection of objects $F(j)$ and a collection of morphisms $F(j) \to F(k)$. 
A cocone on a given diagram is a natural transformation $\lambda:F \to c$ where we abuse notation to write that $c: J \to \cC$ is the constant functor returning the object $c(j) = c \in \cC$ for all $j \in J$. 
We often call the components $\lambda_j:F(j) \to c$ the \emph{legs}, and note that this requirement says that for any $f:j \to k$ in $J$, $\lambda_k \circ F[f] = \lambda_j$. 
A cocone $\lambda:F \to c$ is called a colimit if for any other cocone $\lambda':F \to c'$, there is a unique  morphism $u:c \to c'$ such that 
\begin{equation*}
\begin{tikzcd}
F(j) 
    \ar[rr, "{F[f]}"]
    \ar[dr, "\lambda_j"]
    \ar[ddr, "\lambda_j'"']
&& F(k)
    \ar[dl, "\lambda_k"']
    \ar[ddl, "\lambda_k'"]
\\
& c \ar[d, dashed, "u"] \\
& c'
\end{tikzcd}
\end{equation*}
commutes for all $f:j \to k$ in $J$. 

We will be particularly interested in functors of the form $F:\Open(X) \to \Set$, which can also be called \emph{pre-cosheaves}. 
A pre-cosheaf is a \emph{cosheaf} if it satisfies a gluing axiom,  meaning that $F(U)$ is entirely determined by $F(U_\alpha)$ for any cover $\{U_\alpha\}_\alpha$. 
Specifically, given an open set $U$ and a cover $\{ U_\alpha \mid \alpha \in A\}$ of $U$, define a category $\cU = \{U_\alpha \cap U_{\alpha'} \mid \alpha,\alpha' \in A \}$ with morphisms given by inclusion. 
Then we have a diagram $F:\cU \to \Set$, and as such can consider its colimit $\lambda:F \to L$.
If the unique map $L \to F(U)$ given by the colimit definition is an isomorphism, then $F$ is called a \emph{cosheaf}.

\subsection{Functorial Representation of Embedded Data}
\label{ssec:functorGraphs}

Assume we are given as input a pair of compact topological spaces with valued functions $f: \X \to \R^d$ and $g: \Y \to \R^d$.
We will construct a cover  controlled  by diameter $\delta\geq 0$ of the images $f(\X) \cup g(\Y)$ in order to define the discretized mapper complex. 
Assume that a bounding box containing $f(\X) \cup g(\Y)$ can be written as $[-B,B]^d = [-L\delta, L\delta]^d$. 

Following \cite{Kaczynski2004}, $\delta$ induces a discretization of $[-L\delta, L\delta]^d$ into a cubical complex in the following way. 
An \emph{elementary interval} is an open interval in $\R$ of the form\footnote{Note our use of open intervals here in order to have open sets in our cover later, which differs from the definition given in \cite{Kaczynski2004}.} 
$(\ell\delta, (\ell+1)\delta)$ or a single point viewed as a degenerate interval
$[\ell]:= [\ell\delta, \ell\delta]$ for $\ell \in [-L,\cdots, L] \subset \Z$.
These are called non-degenerate and degenerate intervals, respectively. 
An elementary (open) cube $Q$ is a finite product of $d$ elementary intervals, i.e.
   $ \sigma = I_1 \times I_2 \times \cdots \times I_d \subset [-B,B]^d$. 
The dimension of a cube $\sigma$ is given by the number of intervals used which are non-degenerate. 
This means that 
$0$-cubes are vertices at grid locations 
$(i\delta, j \delta, \ldots, k \delta) \in \delta \cdot \Z^d$, 
$1$-cubes are edges (not including their endpoints), 
$2$-cubes are squares (not including their boundaries),
$3$-cubes are voxels, etc. 
The collection of elementary cubes discretizing $[-B,B]^d$ is a finite cubical complex $K$.
This construction comes with a face relation which gives a poset structure, where we write $\sigma \leq \tau$ iff $\sigma \subseteq \overline{\tau}$, where $\overline{\tau}$ denotes the closure of the set. 
In order to differentiate between the combintorial and continuous settings, we write $|\sigma|$ for the geometric realization  in $\R^d$ of a combinatorial object $\sigma \in K$. 

This complex $K$ induces a cover $\cU$ of $[-B,B]^d$ as follows. 
For any cube $\sigma \in K$, we can find the upper closure of $\sigma$ using the face relation, i.e.~$\sigma^\uparrow = \{ \tau \in K \mid \tau \geq \sigma\}$.
The \emph{cover element associated to $\sigma$} is  $U_\sigma = \bigcup_{\tau \in \sigma^{\uparrow}} |\tau|$.  
Then we write the cover as $\cU = \{U_\sigma \mid \sigma \in K\}$. 
Note that there is a poset relation on $\cU$ given by inclusion, and in particular, $U_\sigma \subseteq U_\tau$ iff $\tau \leq \sigma$.

We next endow the poset $(\cU, \subseteq)$ with the Alexandroff topology, following \cite{Barmak2011}. 
For any set $S \subseteq \cU$, the upper closure, or \emph{upset}, is 
$S^{\uparrow} = \{U \in \cU \mid V \subseteq U, \, V \in S \}$ and the downward closure, or \emph{downset}, is
$S^{\downarrow} = \{ U \in \cU \mid U \subseteq V, \, V \in S\}$. 
We give $(\cU, \subseteq)$ the Alexandroff topology
$\Open(\cU)$, where a set $S\subseteq K$ is open iff the following holds: 
for any $U \in S$ and any $V \subseteq U$, $V \in S$. 
Equivalently, this means that $S$ is its own downset, i.e.~$S = S^\downarrow$.
See Fig.~\ref{fig:NewNotation} for an example of this notation in the case of $d=1$, where the open set $U_{\sigma_i}$ is associated to the point $\sigma_i  = i\delta$, and open set $U_{\tau_i}$ is associated to the edge $\tau_i = (i\delta,(i+1)\delta)$. 

\begin{figure}
    \centering
    \includegraphics[width = .5\textwidth]{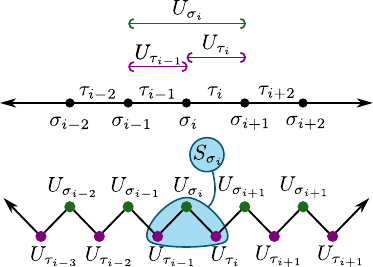}
    \caption{Example of the notation for the case $d=1$, where a discretization of the real line (top) induces a cover, and the poset representation of this cover can be seen in the bottom row. The basic open associated to $U_{\sigma_i}$ is the three element set $S_{\sigma_i} = \{U_{\tau_{i-1}},U_{\sigma_i},U_{\tau_i} \}$.}
    \label{fig:NewNotation}
\end{figure}

It can be checked that this topology has a basis given by the collection  
$\{S_\sigma\}_{\sigma \in K}$ 
where we write 
\begin{equation}
\label{eq:S_sigma}    
    S_\sigma := \{ U_\sigma\}^{\downarrow} 
        = \{U_\tau \in \cU  \mid U_\tau \subseteq U_\sigma \} 
        = \{U_\tau \in \cU \mid \tau \in \sigma^\uparrow\}
\end{equation}
and call $S_\sigma$ a \emph{basic open} set.
This complex is constructed so that for any subset $S \subset \cU$, the geometric realization $|S|:= \bigcup_{U \in S}U \subset \R^d$ is an open set; and further the notation is reasonable since the geometric realization of the basis set associated to $\sigma$ is the same as the open set associated to $\sigma$, i.e.~$|S_\sigma| = U_\sigma$ for all $\sigma \in K$. 
Again, see Fig.~\ref{fig:NewNotation} for an example of this notation.

We assume that the inputs $f:\X \to \R^d$ and $g:\Y \to \R^d$ are semi-algebraic maps defined on semi-algebraic sets in $\mathbb{R}^k$, 
as well as being compact as assumed earlier. Recall that the class of semi-algebraic sets is the smallest class of sets defined by a finite number of polynomial (in)equalities $\{x \in \mathbb{R}^k \mid p(x) \geq 0\}$ that is closed under complement, finite union, and finite intersection. A map $f\colon \X \to \mathbb{R}^d$ is semi-algebraic if the graph of $f$ is a semi-algebraic set in $\mathbb{R}^k \times \mathbb{R}^d$. A semi-algebraic set is \emph{semi-algebraically connected} if it cannot be written as the disjoint union of two non-empty open semi-algebraic sets. Analogously, a semi-algebraic set $X$ is \emph{path connected} if for any two points $x, x' \in X$ there is a continuous semi-algebraic map $\gamma \colon [0,1] \to X$ such that $\gamma(0) = x'$ and $\gamma(1) = y'$. Note that, by definition, a semi-algebraically path connected semi-algebraic set is also path connected. From this observation, combined with Theorem 2.4.5 and Proposition 2.5.13 in \cite{Bochnak1998}, we see that any connected semi-algebraic set is also path-connected. We make this restriction so that we have the following property. 

\begin{lemma}
Given a semi-algebraic map $f: \X \to \R^d$ with $\X$ semi-algebraic, and any semi-algebraic open set $U \subset \R^d$, $\pi_0(f\inv(U))$ is finite. 
\end{lemma}

We refer the reader to \cite{bpr} for an overview of the basic properties of semi-algebraic sets and maps, from which this lemma easily follows.

Because each grid cell itself is a semi-algebraic set, we have that for any open set $S \in \Open(\cU)$, the set of (path) connected components $\pi_0 f\inv(|S|)$ is finite. 

Then we have a functor $F:\Open(\cU) \to \Set$ given by 
\begin{equation*}
    \begin{matrix}
    F: &  \Open(\cU) & \to & \Set\\
    & S & \mapsto & \pi_0 f\inv(|S|).
    \end{matrix}
\end{equation*}
Functoriality of $\pi_0$ means that for $S \subseteq T$, there is an induced map
\[F[S \subseteq T] \colon \pi_0f\inv(|S|) \to \pi_0f\inv(|T|),\]
so that $F$ satisfies the requirements of a functor.
Indeed, this functor is actually a cosheaf so moving forward, we assume that $F$ and $G$ are cosheaves, even if they was not obtained from some input topological space. 
When the sets involved are obvious in the notation, we will  write the induced map as $F[\subseteq]:F(S) \to F(T)$. 

\subsection{Thickenings}
\label{ssec:thickenings}
\begin{figure}
    \centering
    \includegraphics[width = 0.8\textwidth]{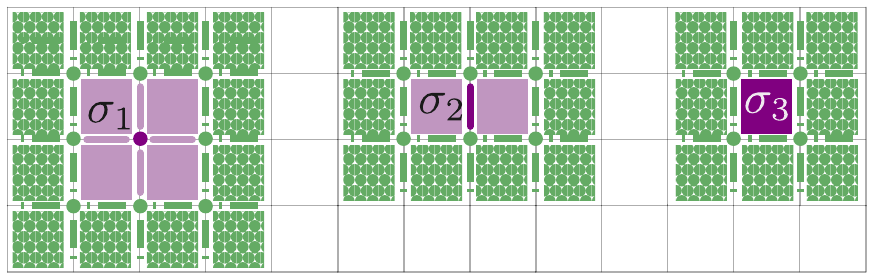}
    \caption{Examples of open sets for $d=2$. The basic open $S_\sigma$ is given in light purple for each cell type shown in dark purple (vertex $\sigma_1$, edge $\sigma_2$ and square $\sigma_3$). The 1-thickening of the open sets are given by including the green portion. }
    \label{fig:thickenings}
\end{figure}
Given any set $S \in \Open(\cU)$, the 1-thickening\footnote{We note that this definition is distinct from the morphological operation of thickening that is a tool in image processing, although quite similar to the concept of dilation from that field.}  is defined by taking the downward closure of the upper closure  of $S$,  written as 
$   S^1 = (S^{\uparrow})^\downarrow$.
This operation can be thought of as taking the star of the closure of the set; see Fig.~\ref{fig:thickenings} for examples. 
The $n$-thickening is defined to be $n$ repetitions of the process given recursively as
\begin{equation*}
    S^n = 
    \begin{cases}
    S & n = 0\\
    (S^{n-1})^{\uparrow\downarrow} & n \geq 1.
    \end{cases}
\end{equation*}
Each $S^n$ is itself an open set in $\Open(\cU)$, and if $S \subseteq T$, then $S^n \subseteq T^n$. 
Thus we can view this operation as a functor on the category $\Open(\cU)$ with morphisms given by inclusion:
\begin{equation*}
    \begin{matrix}
    (-)^n: & \Open(\cU)  & \to &  \Open(\cU)\\
    & S & \mapsto & S^n.
    \end{matrix}
\end{equation*}
In Sec.~\ref{sec:technicalProofs}, 
we show that $(-)^n$ is a functor.
Because of the cosheaf assumption, we also have that $F(S^n) = \varinjlim_{S_\sigma \subset S^n}F(S_\sigma)$. 
Another useful property of this thickening, proved in Sec.~\ref{sec:technicalProofs}, is described in Lem.~\ref{lem:composedthickenings}.
\begin{restatable}{lemma}{composedThickenings}
\label{lem:composedthickenings}
For any $n, n' \geq 0$ and $S \in \Open(\cU)$, 
\[(S^{n})^{n'} = S^{n+n'}.\] 
\end{restatable}

We can use this thickening to build an interleaving distance on cosheaves of the form 
${F:\Open(\cU) \to \Set}$.
The first necessary ingredient is the composition of functors $F \circ ( - )^n: \Open(\cU) \to \Set$, which we  denote by $F^n$. 
This means $F^n(S) = F(S^n)$, followed by a similar setup for $G^n$. 
With this notation, an interleaving is a pair of natural transformations $\phi:F \Rightarrow G^n$ and $\psi:G \Rightarrow F^n$, so a component of $\phi$ is a set-map $\phi_S:F(S) \to G(S^n)$. 
There is another component at $S^n$, $\phi_{S^n}:F(S^n) \to G(S^{2n})$, which can also be viewed as a component of a different natural transformation $\phi^n:F^n \Rightarrow G^{2n}$. 
For this reason, we use the notation $\phi_{S^n}$ and $\phi_S^n$ interchangeably when $\phi$ is indeed a natural transformation.\footnote{We are implicitly using Lem.~\ref{lem:composedthickenings} to write the maps this way.}

We are now ready to introduce our notion of interleaving distance. 
\begin{definition}
\label{def:interleavingDistance}
Given cosheaves $F,G:\Open(\cU) \to \Set$ and $n \geq 0$, an \emph{$n$-interleaving} is a pair of natural transformations 
$\phi:F \Rightarrow G^n$
and 
$\psi:G \Rightarrow F^n$
such that the diagrams 
\begin{equation*}
    \begin{tikzcd}
        F(S) 
            \ar[rr, "{F[S \subseteq S^{2n}]}"]   
            \ar[dr, "\phi_S"',violet] 
            & & F(S^{2n}) & 
        & F(S^n) \ar[dr]
            \ar[dr, "\phi_{S^{n}}",violet]
        & \\
        & G(S^n)\ar[ur, "\psi_{S^n}"', orange]  & & 
        G(S) 
            \ar[rr, "{G[S \subseteq S^{2n}]}"']
            \ar[ur, "\psi_{S}", orange]  
        && G(S^{2n})
    \end{tikzcd}
\end{equation*} 
commute for all $S \in \Open(\cU)$.
The interleaving distance is given by 
\begin{equation*}
    d_I(F,G) = \inf\{ n \geq 0 \mid \text{there exists an $n$-interleaving} \}, 
\end{equation*}
and is set to be $d(F,G) = \infty$ if there is no interleaving for any $n$.
\end{definition}
As shown in Sec.~\ref{sec:technicalProofs}, this definition fits in the framework built by Bubenik et al.~\cite{Bubenik2014a} and thus it is an extended pseudometric. 

\section{Loss Function and Bounds}
\label{sec:loss-function}

In this section, we introduce a loss function for interleavings on $\R^d$-mapper complexes.  
We give the definition of the loss function (Defn.~\ref{def:Loss_v1}) in Sec.~\ref{ssec:LossFunction}, and present our first version of the bound as Thm.~\ref{thm:bound} in Sec.~\ref{ssec:bound_v1}.
However, this version of the bound requires checking diagrams for all possible open sets $S \in \Open (K)$ which creates a combinatorial explosion that is counterproductive in practice.
Thus, in Sec.~\ref{ssec:BasisBound}, we prove this loss function can be replaced with an improved loss function which only needs to check the open sets for a basis of $\Open(\cU)$. 

\subsection{Loss Function Definition}
\label{ssec:LossFunction}
We start by turning each non-empty $F(S)$ (similarly $G(S)$) into a metric space, as follows.
\begin{definition}
Define the distance $d_S^{F}(A,B)$ for $A,B \in F(S)$ to be the smallest $n$ such that $A$ and $B$ represent the same connected component when included into $S^n$. 
Specifically,
\begin{equation*}
    d_S^{F}(A,B) = 
    \min\{ n\geq0 \mid F[S \subset S^n](A) = F[S \subset S^n](B)\}. 
\end{equation*}
If no such $n$ exists, then we  set $d_S^{F}(A,B) = \infty$.
\end{definition}

It is easy to see that this definition satisfies the definition of an extended metric. 
Indeed, it is actually an extended ultrametric since $d_S^F(A,C) \leq \max \{ d_S^F(A,B) , d_S^F(B,C) \}$, although we will not need that additional structure here.

Consider the example of Fig.~\ref{fig:examplegraph-distance} with a single input graph encoded by a cosheaf  $F:\Open(\cU) \to \Set$. 
The set $F(S)$ has two elements, which we denote by $A$ and $B$ as they represent the connected components containing the points $a$ and $b$ respectively. 
Then $d_S^F(A,B) = 1$, since thickening the set $S$ by $1$ puts $a$ and $b$ in the same connected component.
Likewise, denoting the elements of $F(T)$ by $W$ and $Z$, we see that $d_T^F(W,Z) = 2$ since we must expand the set $T$ twice before $w$ and $z$ are in the same connected component. 
\begin{figure}
    \centering
    \includegraphics[width = 0.4\textwidth]{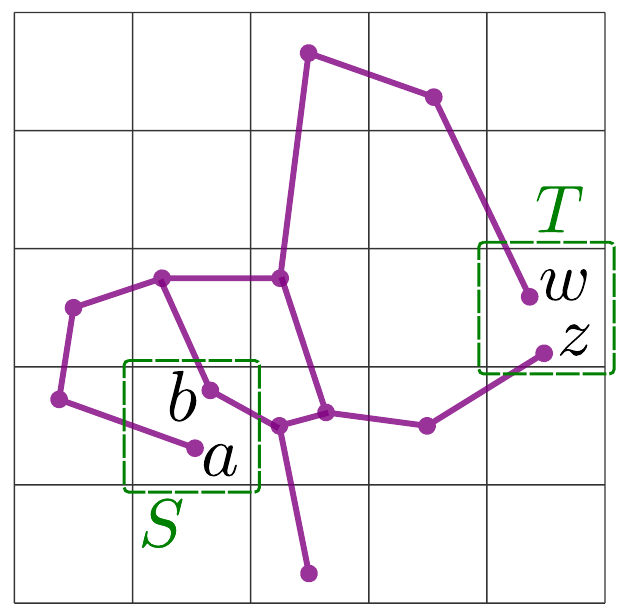}
    \caption{An example geometric graph used to calculate example distances in $d_S^F$ and $d_T^F$. }
    \label{fig:examplegraph-distance}
\end{figure}

As a first useful property of this distance, thickening a set implies that the distance between components will only decrease.
For an example, consider $W,Z \in F(T)$ representing points $w$ and $z$ in Fig.~\ref{fig:examplegraph-distance}.
As noted previously, $d_T^F(W,Z) = 2$. 
However, if the elements ${W',Z' \in F(T^1)}$ represent the connected components in the 1-thickening of $T$, then $d_{T^1}^F(W',Z') = 1$, and in particular, $d_T^F(W,Z) \geq d_{T^1}^F(W',Z')$. 
This idea is formalized in the following lemma:

\begin{lemma}
\label{lem:distanceContraction}
Fix $k \geq 0$ and any $A,B \in F(S)$ with images $A' = F[S\subseteq S^k](A)$ and \linebreak $B' = F[S\subseteq S^k](B)$ in $F(S^k)$. 
Then 
\begin{equation*}
    d_{S^k}^F(A',B') = \max\{0,  d_S^F(A,B) - k \} =
    \begin{cases}
        0 & \text{if } k \geq n\\
        d_S^F(A,B) - k & \text{if }0 \leq k <n
    \end{cases}
\end{equation*}
and in particular, $d_S^F(A,B) \geq d_{S^k}^F(A',B')$. 
\end{lemma}
\begin{proof}
Let $n = d_S^F(A,B)$, so that we know the image of $A$ and $B$ in $F(S^n)$ is the same. 
If $k \geq n$, then we use the functor maps $F(S) \to F(S^n) \to F(S^k)$ to see that the images of $A$ and $B$ are the same in $F(S^n)$ so they are the same in $F(S^k)$. 
Then $d_{S^k}^F(A',B') = 0$. 
If $k < n$, then we have the maps $F(S) \to F(S^k) \to F(S^n)$.
Because we know that $A$ and $B$ do not map to the same thing prior to $n$, we have $d_{S^k}^F(A',B') = n-k$, completing the proof.
\end{proof}

We use this framework as follows: first, assume we are given $F$ and $G$ but our attempts at finding an interleaving  do not necessarily satisfy the requirements of a natural transformation. 
Normally, a natural transformation $\eta:H \Rightarrow H'$ is a collection of component morphisms ${\eta:H(S) \to H'(S)}$ which commute with the inclusions: 
\[
\begin{tikzcd}[sep=scriptsize]
	{H(S)} && {H(T)} \\
	\\
	{H'(S)} && {H'(T)}.
	\arrow["{H'[\subseteq]}", from=3-1, to=3-3]
	\arrow["{\eta_u}"', from=1-1, to=3-1]
	\arrow["{H[\subseteq]}", from=1-1, to=1-3]
	\arrow["{\eta_T}", from=1-3, to=3-3]
\end{tikzcd}
\]
The following definitions, inspired by~\cite{Robinson2020} and~\cite{nlab:unnatural_transformation}, give names to collections of component morphisms used to define an interleaving where the square might not commute. 

\begin{definition}
\label{def:assignment}
Given functors $H,H':\Open(\cU) \to Set$, an \emph{unnatural transformation}\footnote{A  natural transformation is an unnatural transformation which just happens to follow commutativity properties. In other words, natural and unnatural transformations are not mutually exclusive. This vocabulary follows from~\cite{nlab:unnatural_transformation} so we accept no responsibility for the linguistic implications.}   $\eta:H \rightarrow H'$ is a collection of maps $\eta_S:H(S) \to H'(S)$ with no additional promise of commutativity. 

For a fixed $n \geq 0$ and cosheaves $F$ and $G$, an \emph{assignment}, or more specifically an \emph{$n$-assignment}, is a pair of unnatural transformations $\phi:F \Rightarrow G^n$ and $\psi:G \Rightarrow F^n$.

\end{definition}

In order to clarify notation, for the remainder of the paper, we will be using $n$-assignments to  build $(n+k)$-interleavings, which by definition will be required to be natural transformations. 
When the $n$-assignment might not commute, we  denote its maps by lower case $\phi$ and $\psi$;  for $(n+k)$-assignments which are constructed to be natural transformations, we  denote them by $\Phi$ and $\Psi$. 

In addition, we assume for the remainder of the paper that $n$ is large enough for an assignment to exist. 
That is, it is possible that for some given $F(S)$, $G(S^n)$ might be empty for $n$ small enough and thus there is no available map from one to the other. 
However, because we have assumed a compact input, $f(\X)$ and $g(\Y)$ is contained in a compact interval, and thus, we have that the $\sigma$ for which  $F(S_\sigma)$  is not empty is contained in some interval (in the poset sense). 
So long as $n$ is large enough that the Hausdorff distance between the images $f(\X)$ and $g(\Y)$ is at most $\delta n$, $G(S^n)$ will be non-empty for any non-empty $F(S)$ (and vice versa).  

In the spirit of  \cite{Robinson2020}, we measure the quality of a choice of an  $n$-assignment $\phi, \psi$ using the collections of distances $\{d_S^F \mid S \in \Open(\cU)\}$ and $\{d_S^G \mid S \in \Open(\cU)\}$.  
First, note that checking that $\phi$ and $\psi$ are natural transformations means ensuring the diagrams
\begin{equation*}
    \begin{tikzcd}
        F(S)  
            \ar[r, "{F[\subseteq ]}"] 
            \ar[dr, "\phi_S"', very near start, violet]
        & F(T)
            \ar[dr, "\phi_T"', very near start, violet]
        & \\
        & G(S^n) 
            \ar[r, "{G[\subseteq ]}"'] 
        & G (T^n)
    \end{tikzcd}
    \begin{tikzcd}
        & F(S^n)
            \ar[r, "{F[\subseteq ]}"] 
        & F (T^n)\\
        G(S)
            \ar[r, "{G[\subseteq ]}"'] 
            \ar[ur, "\psi_S", very near start, orange]
        & G(T) 
            \ar[ur, "\psi_T", very near start, orange]
        & 
    \end{tikzcd}
\end{equation*}
commute. 
As we use them repeatedly, we will denote these diagrams by $\Parallelograml_\phi(S,T)$ and $\Parallelogramr_\psi(S,T)$, dropping the subscript when it is clear from context.
Then checking whether the pair constitutes an interleaving involves checking commutativity of the diagrams
\begin{equation*}
\label{eq:fourDiagrams}
\begin{tikzcd}
        F(S) 
            \ar[rr, "{F[S \subseteq S^{2n}]}"]   
            \ar[dr, "\phi_S"',violet] 
            & & F(S^{2n}) & 
        & F(S^n) \ar[dr]
            \ar[dr, "\phi_{S^{n}}",violet]
        & \\
        & G(S^n)\ar[ur, "\psi_{S^n}"', orange]  & & 
        G(S) 
            \ar[rr, "{G[S \subseteq S^{2n}]}"']
            \ar[ur, "\psi_{S}", orange]  
        && G(S^{2n})
    \end{tikzcd}
\end{equation*}
which we denote by $\triangled_{\phi,\psi}(S)$ and $\triangleu_{\phi,\psi}(S)$ respectively, again dropping the subscripts when unnecessary. 
We measure quality of the given assignments by checking how far these four diagrams are from commuting in the sense of the distances defined at the terminal vertex of the shape. 

\begin{definition}
\label{def:Loss_v1}
Fix an $n$-assignment
$(\phi,\psi)$. 
We define four \emph{diagram loss functions}: 
\begin{align*}
\Lpl^{S,T}(\phi)
    &= \max\limits_{\alpha \in F(S)} 
       d_{T^n}^{G}(
        \varphi_T \circ F[S \subseteq T](\alpha),
        G[S^n \subseteq T^n] \circ \varphi_S(\alpha)
        )\\
\Lpr^{S,T} (\psi)
    &= \max\limits_{\alpha \in G(S)} 
       d_{T^n}^{F}(
        \psi_T \circ G[S \subseteq T](\alpha), 
        F[S^n \subseteq T^n] \circ \psi_S(\alpha)
    )\\
\Ltd^S (\phi,\psi)
    &= \max\limits_{\alpha \in F(S)}  \Big \lceil \tfrac{1}{2} \cdot d_{S^{2n}}^{F}(
    F[S \subseteq S^{2n}] (\alpha),
    \psi_{S^n} \circ \varphi_S(\alpha)
    ) \Big \rceil\\
\Ltu^S (\phi,\psi)
    &= \max\limits_{\alpha \in G(S)}\Big \lceil \tfrac{1}{2} \cdot d_{S^{2n}}^{G}(
    G[S \subseteq S^{2n}](\alpha),
    \varphi_{S^n} \circ \psi_S(\alpha)
    )\Big \rceil.
\end{align*}
Then the loss for the given assignment is defined to be
\begin{equation*}
L(\phi,\psi) = \max_{S\subseteq T}\left\{\Lpl^{S,T}, \Lpr^{S,T}, \Ltu^S, \Ltd^S\right\}.
\end{equation*}
\end{definition}

These loss functions are defined in a way so that while the diagram in question might not commute, pushing $n$ forward by the loss value will send the elements to the same place. 
For example, if   $\Lpl^{S,T}(\phi)  =k$, then in the diagram 
\begin{equation}
\label{eqn:dgm:parallelExtendK}
\begin{tikzcd}
        F(S)  
            \ar[r, "{F[\subseteq ]}"] 
            \ar[dr, "\phi_S"', very near start, violet]
        & F(T)
            \ar[dr, "\phi_T"', very near start, violet]
        & \\
        & G(S^n) 
            \ar[r, "{G[\subseteq ]}"'] 
        & G (T^n) \ar[r] 
        & G(T^{n+k})
\end{tikzcd}
\end{equation}
the image of a point from $F(S)$ is the same in $G(T^{n+k})$ following both paths. 
Similarly, if $\Ltd^S (\phi,\psi)=k$, then in the diagram 
\begin{equation}
\label{eqn:dgm:triExtendK}
    \begin{tikzcd}
        F(S) 
            \ar[rr, "{F[ \subseteq ]}"]   
            \ar[dr, "\phi_S"',violet] 
            & & F(S^{2n})  \ar[r] 
            & F(S^{2(n+k)})
        \\
        & G(S^n)\ar[ur, "\psi_{S^n}"', orange]  & 
    \end{tikzcd}
\end{equation}
the image of a point in $F(S)$ is the same (following both paths) in $F(S^{2(n+k)})$ even if not in $F(S^{2n})$.

\para{An Example:} 
Consider Fig.~\ref{fig:nonzero_loss_finite} and fix $n=1$. 
Denote the connected component of the point $a$ in $F(S)$, $F(S^1)$, and $F(S^2)$ by $A$, $A'$, and $A''$, respectively.
Similarly, the connected component of the point $b$ is denoted by  $B'' \in G(S^{2})$. 
Follow the same form for the connected components of points $w$ and $z$ in $G$.
The interleaving diagrams can be collected together as 
\begin{equation}
\label{eq:interleavingLadder_example}
    \begin{tikzcd}[row sep=large, column sep=huge]
        {\color{blue}\{A\}}  
            \ar[r, "{F[S \subseteq S^1]}"] 
            \ar[dr, "\phi_S"', very near start, violet]
        & {\color{blue}\{A'\}} 
            \ar[r, "{F[S^1 \subseteq S^{2}]}"] 
            \ar[dr, "\phi_{S^n}"', very near start, violet]
        & {\color{blue}\{A'',B''\}} \\
        {\color{red}\{W,Z\}} 
            \ar[r, "{G[S \subseteq S^1]}"'] 
            \ar[ur, "\psi_S", very near start, orange, crossing over]
        & {\color{red}\{W',Z'\}}
            \ar[r, "{G[S^1 \subseteq S^{2}]}"'] 
            \ar[ur, "\psi_{S^1}", very near start, orange, crossing over]
        & {\color{red}\{W'',Z''\}}
    \end{tikzcd}
\end{equation}
noting that the horizontal maps are determined by sending a letter to the same letter with an additional prime. 
The distances between the points in their respective sets are
\begin{equation*}
    \begin{matrix}
    &&& d_{S^2}^F(A'',B'') = 1; \\
    \\
    & d_S^G(W,Z) = 3, & 
    d_{S^1}^G(W',Z')  = 2, & 
    d_{S^2}^G(W'',Z'') = 1. 
    \end{matrix}
\end{equation*}

\begin{figure}
    \centering
    \includegraphics[width = 0.4\textwidth]{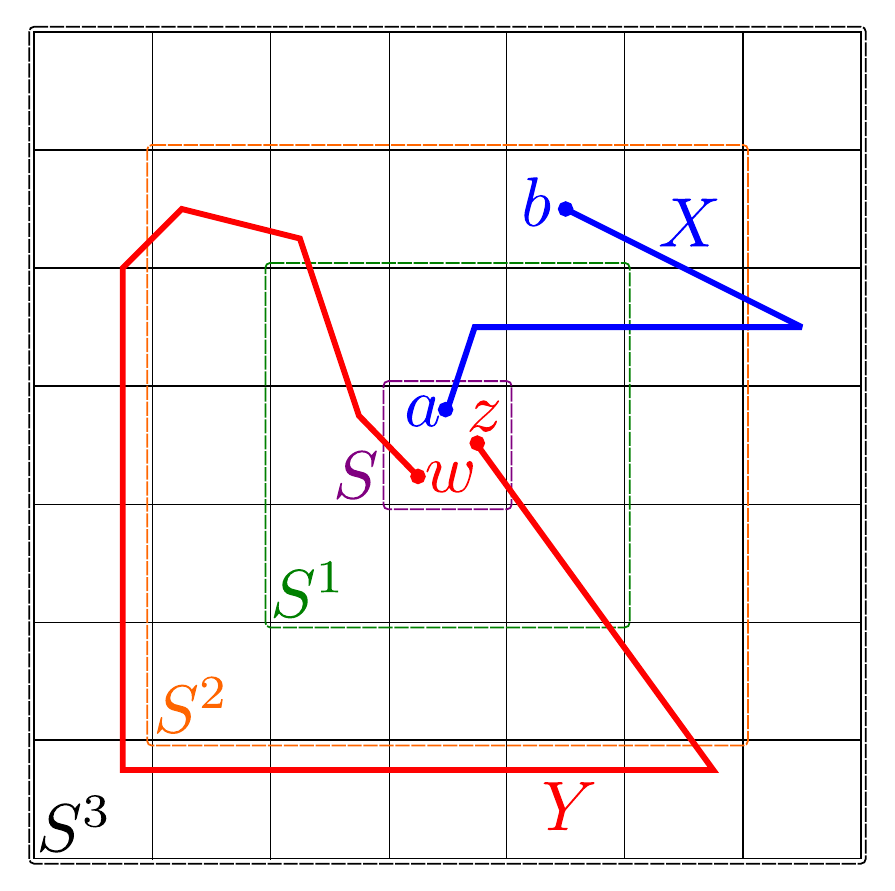}
    \caption{An example of two input geometric graphs, $X$ and $Y$.}
    \label{fig:nonzero_loss_finite}
\end{figure}
Consider the following example assignment:
\begin{equation*}
\begin{matrix}
    \phi_S: A \mapsto W',& & 
    \psi_S: W,Z \mapsto A',\\
    \phi_{S^1}:A' \mapsto W'',& &
    \psi_{S^1}: \substack{W' \mapsto A'' \\ Z' \mapsto B'' }.
\end{matrix}
\end{equation*}
In this case, we then have that 
$\Lpl^{S,S^n} = 0$,
$\Lpr^{S,S^n} = 1$, 
$\Ltd^S = 0 $, 
and $\Ltu^S= 1$, 
so again $L(\phi,\psi) \geq 1$.
For this particular example, no $n=1$ interleaving is possible so any choice of assignment will have a non-zero loss (the easiest check is to see that any choice of assignment will force $\Ltu^S =1$). 

\subsection{Bounding the Interleaving Distance}
\label{ssec:bound_v1}

We now use the loss function to give an upper bound for the interleaving distance.  
\begin{restatable}{theorem}{FirstLossBound}
\label{thm:bound} 
    For an $n$-assignment,  $\phi\colon F \Rightarrow G^n$ and $\psi\colon G \Rightarrow F^n$, 
    \begin{equation*}
        d_I(F, G) \leq  n+L(\phi, \psi).  
    \end{equation*}
\end{restatable}

To prove this, we require the following technical lemma, proved in Sec.~\ref{sec:technicalProofs}.

\begin{restatable}{lemma}{lossimpliescommutes}
\label{lem:lossimpliescommutes}
Assume we are given an $n$-assignment
$\phi:F \Rightarrow G^n$ and 
$\psi:G \Rightarrow F^n$. 
For a fixed $k$, define $(n+k)$-assignments
$\Phi_S = G[S^n\subseteq S^{n+k}]\circ \phi_S$
and 
$\Psi_S = F[S^n\subseteq S^{n+k}]\circ \psi_S$ for all $S \in \Open(\cU)$. 
Then the following hold:
\begin{enumerate}
    \item $\Lpl^{S,T}(\phi) \leq k$ implies $\Parallelograml_{\Phi}(S,T)$ commutes, and thus $\Lpl^{S,T}(\Phi) = 0$.
    \item $\Lpr^{S,T}(\psi) \leq k$ implies $\Parallelogramr_{\Psi}(S,T)$ commutes, and thus $\Lpr^{S,T}( \Psi) = 0$.
    \item 
    $\Ltd^{S}(\phi,\psi) \leq k$ 
    and
    $\Lpr^{S^n,S^{n+k}}(\psi) \leq k$     imply $\triangled_{\Phi, \Psi}(S)$ commutes, and thus $\Ltd^{S}(\Phi, \Psi) = 0$.
    \item $\Ltu^{S}(\phi,\psi) \leq k$  and
    $\Lpl^{S^n,S^{n+k}}(\phi) \leq k$ 
    imply $\triangleu_{\Phi, \Psi}(S)$ commutes, and thus $\Ltu^{S}(\Phi, \Psi) = 0$.
\end{enumerate}
In particular, if  $\phi$ and $\psi$ have $L(\phi,\psi) = 0$, then $\phi$ and $\psi$ constitute an interleaving, and so $d_I(F,G) \leq n$.
\end{restatable}

\begin{proof}[Proof of Thm.~\ref{thm:bound} ]
Set $k = L(\phi,\psi)$,
so by definition,  $\Lpl^{S,T}(\phi) \leq k$, $\Lpr^{S,T}(\psi) \leq k$, $\Ltd^{S}(\phi,\psi) \leq k$, and $\Ltu^{S}(\phi,\psi) \leq k$. 
As in Lem.~\ref{lem:lossimpliescommutes}, construct two $(n+k)$-assignments: 
$\Phi$ given by 
$\Phi_S = G[S^n \subseteq S^{n+k}] \circ \phi$,  and
$\Psi$ given by 
$\Psi_S = F[S^n \subseteq S^{n+k}] \circ \psi$.
By Lem.~\ref{lem:lossimpliescommutes}, this means the diagrams 
$\Parallelograml_{\Phi}(S,T)$,
$\Parallelogramr_{\Psi}(S,T)$,
$\triangled_{\Phi, \Psi}(S)$, and 
$\triangleu_{\Phi, \Psi}(S)$ 
commute for all pairs $S\subseteq T$. 
This implies that $\Phi$ and $\Psi$ are an $(n+k)$-interleaving, giving the theorem. 
\end{proof}

First, notice that this proof works by explicitly constructing an interleaving from a given $n$-assignment.
Second, we have no reason to believe that this bound is tight.
In particular, in Sec.~\ref{ssec:BasisBound} we improve the  bound by way of restricting the computation to the basis for the topology of $K$ but even that is depending on input quality and gives no guarantee.

We include one additional note about when this loss function can be promised to be finite. 
Define the diameter of a metric space to be the largest distance between points, which we denote by 
$
    \mathrm{diam}(X,d) = \sup \{ d(a,b) \mid a,b \in X\}
$,
and note that here, the $\sup$ can be replaced with a $\max$ since we are working in finite metric spaces.
To simplify statements, we define the diameter of the empty set to be zero.
\begin{lemma}
The loss function for an $n$-assignment $(\phi,\psi)$ is bounded above; specifically,
\begin{align*}
    L(\phi,\psi) \leq 
    \max \Bigg( &
    \left\{\mathrm{diam}(F(S^{k}),d_F^{S^{k}}) \mid S \in \Open(\cU), k \in \{ n, 2n\}\right\} \\
    &\cup \left\{\mathrm{diam}(G(S^{k}),d_G^{S^{k}})\mid S \in \Open(\cU), k \in \{ n, 2n\}\right\}
    \Bigg) .
\end{align*}
In particular, if the inputs come from $f:\X\to\R$ and $g:\Y\to\R$ with both $\X$ and $\Y$ connected, then $L(\phi,\psi)$ is finite.
\end{lemma}

\begin{proof}
The parallelogram portions of the loss function $\Lpl$ and $\Lpr$ take values from distances in $F(S^n)$ and $G(S^n)$. 
The triangle portions $\Ltd$ and $\Ltu$ take values  from distances in $F(S^{2n})$ and $G(S^{2n})$. 
So, the maximum for the loss function must be attained on one of these sets, giving the inequality.
For the second statement, if the input graphs each have a single connected component, then any pair of elements $a, b \in F(S)$ map to the same element under the inclusion $F(S) \to F(S^K)$ for a large enough $K$. 
This in turn implies that the diameter of $d_S^F$ is finite for every $S$. 
\end{proof}

\begin{figure}
    \centering
    \includegraphics[width = 0.4\textwidth]{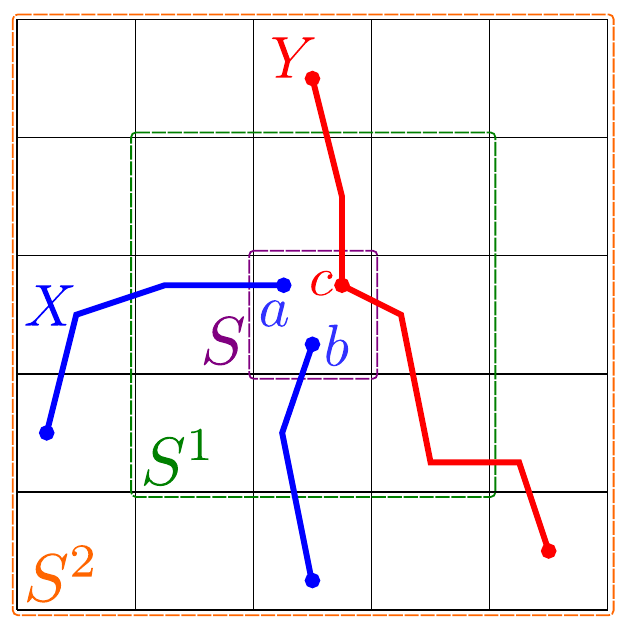}
    \caption{An example of the comparison of two geometric graphs with different numbers of connected components. 
    In this case, because $X$ has one connected component and $Y$ has two, the loss function will be infinite for any assignment. }
    \label{fig:infiniteLoss}
\end{figure}
Consider the example in Fig.~\ref{fig:infiniteLoss}.
Let $\{A,B\}$, $\{A',B'\}$, and $\{A'',B''\}$ be the representatives of the connected components of the points $a$ and $b$ in $F(S)$, $F(S^1)$ and $F(S^2)$ respectively. 
Because there is no $n$ for which the two points are the same connected component of $X$, the distance between $A$ and $B$ is $\infty$ in all three sets. 
Then no matter the choice of $1$-assignment, $\Ltd = \infty$, making the loss function infinite. 

\subsection{Restriction to Basis Elements}
\label{ssec:BasisBound}
We have so far measured the loss function by studying all possible open sets $S$. 
While this is helpful for definitions, it does not make for a reasonable computational setting. 
To that end, we now focus on a basis of the topology, and prove that this basis suffices.

Recall that an open set $S_\sigma\in \Open(\cU)$ (Eqn.~\eqref{eq:S_sigma}) given by the downset of $U_\sigma$ for some $\sigma \in K$ is called a \emph{basic open set}. 
Note that  $\{S_\sigma \mid \sigma \in K \}$ is a basis for the Alexandroff topology. 
We next give a name to the case where we are only given $n$-assignment information for basis elements, or equivalently, if we are given a full assignment but ignore the maps for non-basis open sets.
\begin{definition}
A \emph{basis unnatural transformation} for functors $H$ and $H'$ is a collection of maps $\eta_{S_\sigma}:H(S_\sigma) \to H'(S_\sigma)$ for all basis elements $S_\sigma$ from $\sigma \in K$. 
A \emph{basis $n$-assignment} (or simply a basis assignment) is a pair of basis unnatural transformations
$$
\{\phi_{S_\sigma} :F(S_\sigma) \to G(S^n_\sigma) \mid \sigma \in K\} 
\qquad \text{and}\qquad 
\{\psi_{S_\sigma} :G(S_\sigma) \to F(S^n_\sigma) \mid \sigma \in K\} 
$$
\end{definition}

In this section, we  prove that we can focus our loss function efforts on only those diagrams associated to basic opens, and the solution can be extended to any open set.
\begin{definition}
\label{def:basisLoss}
    The \emph{basis loss function} is defined to be 
\begin{equation*}
L_B(\phi,\psi) = \max_{\sigma \leq \tau}
\left\{
\Lpl^{S_\tau, S_\sigma}, \Lpr^{S_\tau, S_\sigma}, \Ltu^{S_\sigma}, \Ltd^{S_\sigma}
\right\}.
\end{equation*}
\end{definition}
It is immediate from the definitions that $L_B \leq L$ as the $L_B$ maximum is taken over a subset of those used to determine $L$. 
This means, in particular, that if $L=0$ then $L_B = 0$. 
These values are not always equal; for instance, we might have chosen a basis assignment for which every diagram commutes (making $L_B = 0$), but $\phi_T$ defined on non-basis elements causes a non-zero loss function so $L >0$. 
However in the special case where $L_B = 0$, and thus the basis open diagrams are commutative, we do have the ability to extend the information checked to a full interleaving. 
This can be seen in the following lemma, proved in Sec.~\ref{sec:technicalProofs}. 

\begin{restatable}{lemma}{extendToNatTrans}
\label{lem:extendToNatTrans}
Given a basis unnatural transformation
\begin{equation*}
\{\Phi_{S_\sigma}: F(S_\sigma) \to G(S_\sigma^N) \mid \sigma \in K\} 
\end{equation*}
with $\Lpl^{S_\tau, S_\sigma} = 0$ for all $\sigma \leq \tau$, we can extend this to a full natural transformation $\Phi$; i.e.~we can define $\Phi_S$ for all $S$ such that $\Lpl^{S,T} = 0 $ for all $S \subseteq T$. 
\end{restatable}

Note that the symmetric version extending a basis unnatural transformation $\Psi$ to a natural transformation $\Psi: G \Rightarrow F^N$ is obtained in exactly the same way. 
Next, we can take these natural transformations and ensure the triangles commute (thus giving an interleaving) by only checking the basis set triangles, again proved in Sec.~\ref{sec:technicalProofs}.

\begin{restatable}{lemma}{extendTriangles}
\label{lem:extendTriangles}
    Given natural transformations $\Phi:F \Rightarrow G^N$ and $\Psi:G^N \Rightarrow F$ such that $\Ltd^{S_\sigma} = 0$ for all $\sigma \in K$, then $\Ltd^{S} = 0$ for all open sets $S$. 
\end{restatable}

Taken together, we immediately have the following proposition. 
\begin{proposition}
\label{prop:zeros}
Fix a basis $N$-assignment $(\Phi,\Psi)$. 
If $L_B(\Phi,\Psi) = 0$, then $\Phi$ and $\Psi$ can be extended to natural transformations with $L(\Phi,\Psi) = 0$, and thus constitute an interleaving. 
\end{proposition}

Finally, we arrive at our main theorem, where we can use the provided basis $n$-assignment and the calculated loss function to give a bound for the interleaving distance.

\begin{theorem}
\label{thm:secondBound}
Given a basis $n$-assignment  
\begin{equation*}
\phi = \{\phi_{S_\sigma} \mid \sigma \in K\} 
\text{ and } 
\psi = \{\psi_{S_\sigma} \mid \sigma \in K\}, 
\end{equation*}
we have 
\begin{equation*}
    d_I(F,G) \leq n + L_B(\phi,\psi).
\end{equation*}
\end{theorem}

\begin{proof}
This proof proceeds in the same way as that of Thm.~\ref{thm:bound} with some minor modifications of input assumptions. 
First, let $k = L_B(\phi,\psi)$; and 
define a basis $(n+k)$-assignment by
\begin{equation*}
\{\Phi_{S_\sigma} = G[\subseteq] \circ \phi_{S_\sigma} \mid \sigma \in K\}
\qquad \text{ and } \qquad 
\{\Psi_{S_\sigma} = F[\subseteq] \circ \psi_{S_\sigma} \mid \sigma \in K\}. 
\end{equation*}
By Lem.~\ref{lem:lossimpliescommutes}, we know that  $\Lpl^{S_\tau,S_\sigma}(\Phi)  =0$
and
$\Lpr^{S_\tau,S_\sigma}(\Psi)  =0$
for all $\tau \leq \sigma$.
Then by Lem.~\ref{lem:extendToNatTrans}, we can extend $\Phi$ and $\Psi$ to full natural transformations defined for all $S \in \Open(\cU)$. 

To show that $\Phi$ and $\Psi$ constitute an $(n+k)$-interleaving, we must check triangles; i.e.~ensure that $\Ltd^{S}(\Phi, \Psi) = \Ltu^{S}(\Phi, \Psi)= 0$. 
With the goal of using part 3 of Lem.~\ref{lem:lossimpliescommutes}, first note that $\Ltd^{S_\sigma} (\phi,\psi) \leq k$ for basis elements. 
We can see that $\Lpr^{S_\sigma^n, S_\sigma^{n+k}} \leq k$ by using the (non-commutative) diagram 
\begin{equation*}
\begin{tikzcd}
    & F(S_\sigma^{2n}) 
        \ar[rr , "{F[\subseteq]}"] 
    && F(S_\sigma^{2n+k}) 
        \ar[r, "{F[\subseteq]}"]  
    & F\left(S_\sigma^{2(n+k)}\right)\\
    G(S_\sigma^n) 
        \ar[rr, "{G[\subseteq]}"'] 
        \ar[ur, "\psi_{\bullet}"] 
        \ar[urrr, "\Psi_{\bullet}", orange]
    && G(S_\sigma^{n+k}). 
        \ar[ur, "\psi_{\bullet}", very near start] 
        \ar[urr, "\Psi_{\bullet}"', orange]
\end{tikzcd}
\end{equation*}
The leftmost and rightmost triangles commute by definition of $\Psi$, and the orange parallelogram commutes because $\Psi$ is a natural transformation. 
Then chasing any $x \in G(S_\sigma^n)$ up to the top right $F\left(S_\sigma^{2(n+k)}\right)$ results in the same element, giving the required bound on $\Lpr^{S_\sigma^n, S_\sigma^{n+k}} $. 
Using Lem.~\ref{lem:extendTriangles} for $\Phi$ and $\Psi$, $\Ltd^{S} (\Phi,\Psi) =0$ for all open sets $S$. 
The proof that $\Ltu^{S} (\Phi,\Psi) =0$ is similar.
Therefore $\Phi$ and $\Psi$ are an $(n+k)$-interleaving, giving the bound.
\end{proof}

What is surprising about this bound is that despite checking fewer open sets, the loss function for $L_B$ is actually lower than that found using $L$. 
One reason for this is that when we work with the smaller set of input maps, we extend the collection to a ``better'' full assignment, potentially getting rid of some of the causes of a nonzero loss function in the first place. 
For example, a full assignment would be required to provide a map $\phi_S$ for a $S$ with multiple connected components, say $S = T_1 \cup T_2$.
Since no requirements were made of this map based on the $\phi_{T_1}$ and $\phi_{T_2}$ maps, there is a reasonable chance that the loss function contribution from the $\Lpl^{T_1,S}$ is higher than necessary. 
However, in the basis version, we can build the best possible $\phi_T$ given the information over $\phi_{S_1}$ and $\phi_{S_2}$, providing a potentially better, but certainly no worse, bound.

\section{Computation}
\label{sec:computation}

In this section, we show that given an $n$-assignment $\phi$  and $\psi$, we can compute the loss function $L_B(\phi,\psi)$ in polynomial time. 
For simplicity, we describe the algorithm explicitly in the case where $d=1$ for clarity of exposition, before addressing the run time in higher dimensions. 

\subsection{Data Structures for $d=1$}
\label{ssec:DataStructures}
In this section, we build a pair of graphs representing a pair of input functors $F,G$, and use pointers between the vertex and edge sets to represent a given $n$-assignment $\phi$ and $\psi$.
For ease of exposition, we will carefully focus on the case $d=1$, following the example of Fig.~\ref{fig:DataStructureExample} to illustrate our construction. 
While we acknowledge an overuse of the term ``graph'' throughout this paper as it is used in many different contexts, for this section we use the term graph to mean a network; i.e.~a combinatorial pair consisting of vertices and edges. 
Throughout, we denote edges interchangeably by either $(x,y)$ or $xy$ depending on the context and notational complexity. 

At a high level, we will construct graphs for $F$ and $G$, which we denote by $(V_F,E_F)$ and $(V_G, E_G)$.
Then we will build data structures to encode the natural transformations $\phi$ and $\psi$.
For clarity, we use $\phitt$ and $\psitt$ to denote a collection of pointers that will store the information of $\phi$ and $\psi$, respectively.
These will be viewed as set maps\footnote{This notation is meant to imply that $\phitt$ is composed of two maps,  $\phitt:V_F \to V_G$ and $\phitt:E_F \to E_G$.} $\phitt:(V_F,E_F) \to (V_G,E_G)$ and $\psitt:(V_G,E_G) \to (V_F,E_F)$, which map each vertex to a vertex in the other graph and each edge to an edge in the other graph. 
We give explicit constructions of $\phitt$ and $\psitt$ as well as further restrictions on allowable maps in what follows.

Focusing on $d=1$,  the complex $K$, i.e.,~the discretization of $\R$, consists of vertices which we write as $\sigma_{-L},\cdots,\sigma_L$ with heights in our bounding box $[-L\delta,L\delta]$, and with edges $\tau_j = (\sigma_j, \sigma_{j+1})$. 
Then we construct the graph for $F:\Open(\cU) \to \Set$ by generating a vertex for every object in every $F(S_{\sigma_i})$ and connect them using the morphisms of the functor.
This results in a vertex set 
$$V_F = \coprod_{i =1}^B F(S_{\sigma_i}),$$
and an edge for every object in every $F(S_{\tau_i})$, giving edge set 
$$E_F = \coprod_{i =1}^{B-1} F(S_{\tau_i}).$$ 
The endpoints of any edge $e \in F(S_{\tau_i}) \subseteq E_F$ can be found via the attaching maps: 
\begin{align*}
F[S_{\tau_i} &\subseteq S_{\sigma_{i}}](e) \in F(S_{\sigma_i}) 
    \text{ and} \\
F[S_{\tau_i} &\subseteq S_{\sigma_{i+1}}](e) \in F(S_{\sigma_{i+1}}). 
\end{align*}
For example,  $e = (v_4,v_6) \in F(S_{\tau_4})$ in Fig.~\ref{fig:DataStructureExample} has endpoints $v_6 \in F(S_{\sigma_4})$ and $v_4 \in F(S_{\sigma_5})$. 
We store this data in a standard adjacency list.
In addition, each vertex also keeps track of its height, so a vertex $v \in F(S_{\sigma_i})$ also stores the value $i$ as a representation of its height.

\begin{figure}
    \centering
\begin{minipage}{0.65\textwidth}
    \includegraphics[scale=0.5]{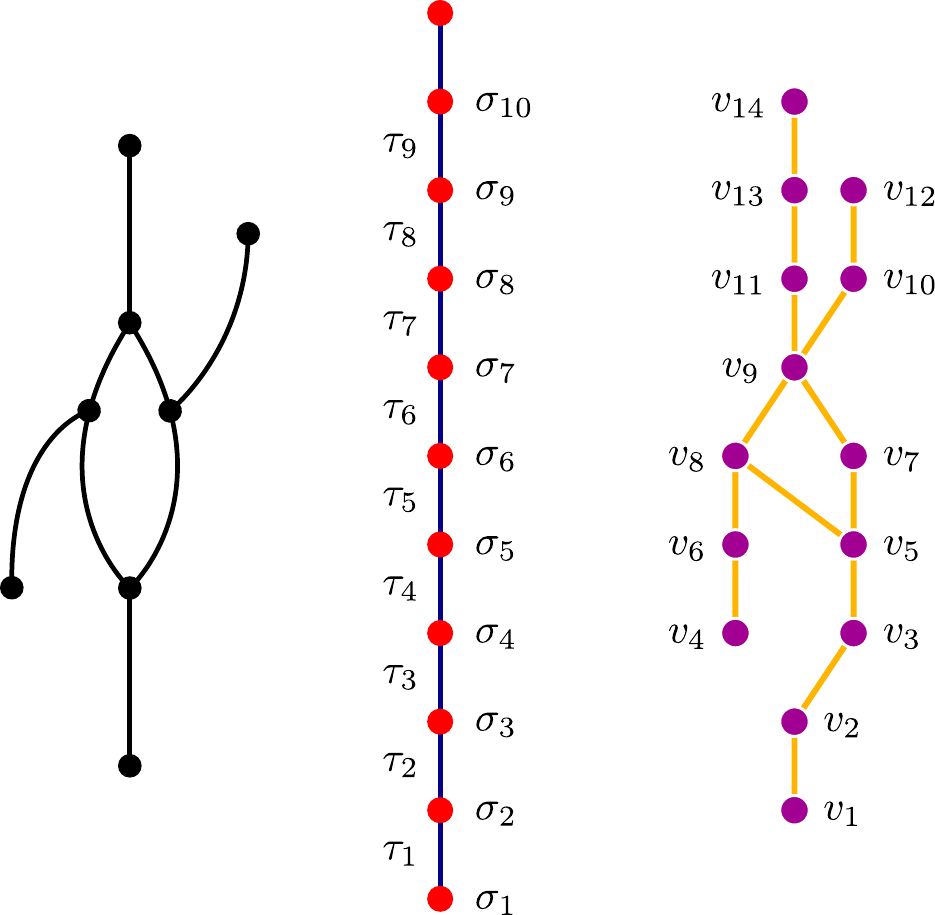}
\end{minipage}
\begin{minipage}{0.3\textwidth}
\begin{align*}
    v_{14}&: \{v_{13}\}, \hspace{3mm} \sigma_{10} \\
    v_{13}&: \{v_{11}, v_{14}\}, \hspace{3mm} \sigma_9 \\
    v_{12}&: \{v_{10}\}, \hspace{3mm} \sigma_9 \\
    v_{11}&: \{v_9, v_{13}\}, \hspace{3mm} \sigma_8 \\
    v_{10}&: \{v_9, v_{12}\}, \hspace{3mm} \sigma_8 \\
    v_9&: \{v_7, v_8, v_{10}, v_{11}\}, \hspace{3mm} \sigma_7 \\
    v_8&: \{v_5, v_6, v_9\}, \hspace{3mm} \sigma_6 \\
    v_7&: \{v_5, v_9\}, \hspace{3mm} \sigma_6 \\
    v_6&: \{v_4, v_8\}, \hspace{3mm} \sigma_5 \\
    v_5&: \{v_3, v_7, v_8\}, \hspace{3mm} \sigma_5 \\
    v_4&: \{v_6\}, \hspace{3mm} \sigma_4 \\
    v_3&: \{v_5, v_2\}, \hspace{3mm} \sigma_4 \\
    v_2&: \{v_3, v_1\}, \hspace{3mm} \sigma_3 \\
    v_1&: \{v_2\}, \hspace{3mm} \sigma_2 \\
\end{align*}   
\end{minipage}
    \caption{From left to right: an example input Reeb graph ($d=1$), a discretization of $\R$, the generated mapper graph, and the data structure encoding the mapper graph.}
    \label{fig:DataStructureExample}
\end{figure}

Next, we encode the information for an assignment $(\phi,\psi)$ between $F,G: \Open(\cU) \to \Set$ by constructing the set maps $\phitt$ and $\psitt$ using the graphs $(V_F,E_F)$ and $(V_G,E_G)$.
Assume we are given $n$-assignments $\phi$ and $\psi$. 
We start by focusing on the vertex set. 
For this, we need to represent the map
$\phi_{S_{\sigma_i}}:F(S_{\sigma_i}) \to G(S_{\sigma_i}^n)$. 
The elements of $F(S_{\sigma_i})$ are already given as vertices, however the elements of $G(S_{\sigma_i}^n)$ are not.
But, because of the cosheaf structure of $G$, the elements of $G(S_{\sigma_i}^n)$ can be seen as the connected components of particular subsets of the $(V_G,E_G)$. 

We emphasize that we are using the term ``subset" because the resulting objects will not be subgraphs in the usual sense. 
These subsets will consist of a subset of the vertex set and a subset of the edge set $(V', E')$, but without the promise that both endpoints of an edge are included in the set. 
However, we can still define connected components in this setting to consist of vertices and edges which can be connected by paths. 

For some vertex $\sigma_i$ of $K$, let 
\[
V_{G,\sigma_i,n} = \{ v \mid v \in G(S_{\sigma_j}), j \in [i-n,i+n]\}
\]
and 
\[
E_{G,\sigma_i,n} = \{e \mid e \in G(S_{\tau_j}), j \in [i-n-1,i+n] \}
\]
and write $(V_G,E_G)_{\sigma_i,n}:=(V_{G,\sigma_i,n}, E_{G,\sigma_i,n})$.
Similarly for the edges of $K$, we can define 
\[
V_{G,\tau_i,n} = \{ v \mid v \in G(S_{\sigma_j}), j \in [i-n+1,i+n]
\]
and 
\[
E_{G,\tau_i,n} = \{e \mid e \in G(S_{\tau_j}), j \in [i-n,i+n] \}
\]
with notation $(V_G,E_G)_{\tau_i,n}:=(V_{G,\tau_i,n}, E_{G,\tau_i,n})$. 
Note that because of the endpoints, these are not induced subgraphs; see Fig.~\ref{fig:Assignment} for examples.
We call either $(V_G,E_G)_{\tau_i,n}$ or $(V_G,E_G)_{\sigma_i,n}$ a \emph{slice} of the graph since they are each a portion of the graph which gives the connected components over an open interval as seen in the following lemma. 

\begin{lemma}
The elements of $G(S_{\sigma_i}^n)$ (resp.~$G(S_{\tau_i}^n)$) are in one-to-one correspondence with the connected components of the subset of the graph $(V_G,E_G)_{\sigma_i,n}$ (resp.~$(V_G,E_G)_{\tau_i,n}$). 
\end{lemma}
\begin{proof} 
This lemma is immediate from noting that $G(S_{\sigma_i}^n)$ is the colimit of the diagram given by $G(S_\sigma)$ for $S_\sigma \subset S_{\sigma_i}^n$ where $\sigma$ is taken over all cells  in $K$ (of both dimension 0 and 1), and then carefully tracking indices of these cells from the above notation.
The edge version is similar. 
\end{proof}

With this, we can return to storing a given $n$-assignment $\phi$ and $\psi$. 
To store the unnatural transformation $\phi$, for each $v \in F(S_{\sigma_i})$, we choose a vertex $\phitt(v) \in V_{G,\sigma_i,n}$, where $\phitt(v)$ is in the connected component of $(V_{G}, E_G)_{\sigma_i,n}$ represented by $\phi_{S_{\sigma_i}}(v) \in G(S_{\sigma_i}^n)$.
We note that given a collection of choices of vertices 
$$\{ \phitt(v) \in G(S_{\sigma_j})\subset V_G \mid v \in F(S_{\sigma_i}) \subset V_F,\, |i-j| \leq n\}$$
and edges 
$$\{ \phitt(e) \in G(S_{\tau_j}) \mid e \in F(S_{\tau_i}) \subset V_F,\, |i-j| \leq n\}$$
we can immediately reconstruct an unnatural transformation $\phi$  by setting $\phi_{S_{\sigma_i}}(v)$ to be the element of $G(S_{\sigma_i}^n)$ representing the connected component of $(V_G,E_G)_{\sigma_i,n}$ containing  $\phitt(v)$. 
As these processes are inverse of each other and the parallel setup can be done for $\psi$ and $\psitt$, we are justified in using this data structure to represent the assignment.

For an example, consider Fig.~\ref{fig:Assignment} where we assume $n=1$. 
If  $\phitt(b) = w$, then $\phi_{S_{\sigma_i}}(b)$ is the connected component that includes $w$ of $(V_G,E_G)_{\sigma_i,1}$ as shown on the right.  
We can similarly find the edge map $\phitt(e)$ for $e \in F(S_{\tau_i})$ by setting it to be an edge in $E_{G,\tau_i,n}$ representing the connected component of $\phi_{S_{\tau_i}}(e) \in G(S_{\tau_i}^n)$ in $(V_G,E_G)_{\tau_i,n}$.
So, for example, in Fig.~\ref{fig:Assignment} where $n=1$, the input data might have $\phitt(ab) = (xy) \in E_G$ and $\phitt(bc) = (uv) \in E_G$.

\begin{figure}
    \centering
    \includegraphics[width = .45\textwidth, align = c]{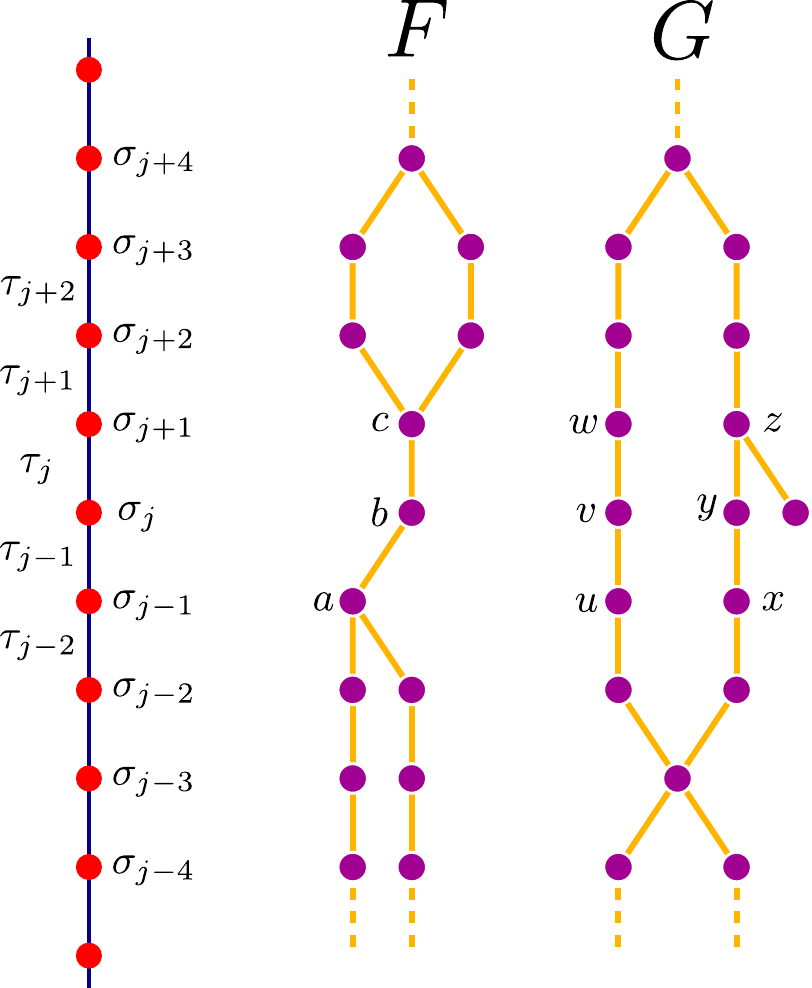} 
    \qquad 
    \begin{minipage}{.25\textwidth}
    \begin{equation*}
    \phitt:  
    \begin{cases}
        b \mapsto w \\ 
        c \mapsto z\\
        ab \mapsto xy\\
        bc \mapsto uv\\
        \phantom{xxl} \vdots 
    \end{cases}
    \end{equation*}
    \begin{equation*}
    \psitt:  
    \begin{cases}
        x \mapsto b \\ 
        w \mapsto c\\
        \phantom{xxl} \vdots
    \end{cases}
    \end{equation*}
    \end{minipage}
    \qquad 
    \includegraphics[width = 0.15\textwidth,align = c]{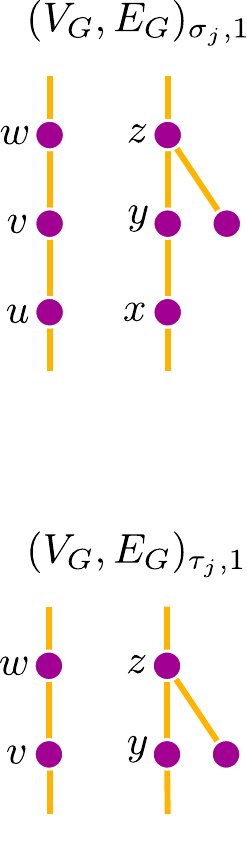}
    \caption{At left are two example input mapper graphs. In the middle are a subset of an example {\phitt} and {\psitt}  input assignment used throughout the text. At right are two example slices of the graphs used when checking connectivity.}
    \label{fig:Assignment}
\end{figure}

\subsection{Algorithm and Complexity for $d=1$}
\label{ssec:Complexity}
In this section, we discuss the complexity of determining $L_B(\phi,\psi)$ given $\phitt$ and $\psitt$. 
First, we will proceed using a binary search on $k \in [0,\cdots, 2L]$ where the maximum is determined by the diameter of the bounding box. 
For a fixed $k$, we will determine if $L_B(\phi,\psi) \leq k$ by checking if $\Lpl^{S_\tau, S_\sigma}$,  $\Ltd^{S_\sigma}$, 
$\Lpr^{S_\tau, S_\sigma}$ and $\Ltu^{S_\sigma}$ are all less than $k$ for all $\sigma$ and $\tau$ in $K$. 
We will describe the cases for 
$\Lpl^{S_\tau, S_\sigma}$ and $\Ltd^{S_\sigma}$, as 
$\Lpr^{S_\tau, S_\sigma}$ and $\Ltu^{S_\sigma}$ are symmetric.

Start with $\Lpl^{S_\tau, S_\sigma}$ and note that in the case where $d=1$, there are two pairs necessary to check for each edge: $\tau_j, \sigma_j$ and $\tau_j,\sigma_{j+1}$.
Fix $\sigma_\ell$ to be either $\sigma_j$ or $\sigma_{j+1}$.
For each edge $e \in F(S_{\tau_j})$, we need to check if the two possible images in $G(S_{\sigma_\ell}^{n+k})$ under the diagram
\begin{equation}
\label{eqn:dgm:parallel_extend_basis}
\begin{tikzcd}
        F(S_{\tau_i})  
            \ar[r, "{F[\subseteq ]}"] 
            \ar[dr, "\phi_{S_{\tau_i}}"',  violet]
        & F(S_{\sigma_\ell})
            \ar[dr, "\phi_{S_{\tau_i}^n}",  violet]
        & & e \ar[r,mapsto] \ar[dr, mapsto]
        & v \ar[dr, mapsto] \\
        & G(S_{\tau_i}^n) 
            \ar[r, "{G[\subseteq ]}"'] 
        & G (S_{\sigma_\ell}^n) \ar[r] 
        & G(S_{\sigma_\ell}^{n+k})
        & \substack{\\{[e']}} \ar[r, mapsto, shift right] 
        & \substack{[w]\\{[e']}} \ar[r, mapsto, shift left] \ar[r, mapsto, shift right] 
        & \substack{[w]\\{[e']}}
\end{tikzcd}
\end{equation}
are the same. 
Note that we use $[-]$ to  represent the connected component in the relevant slice of the graph containing that edge or vertex. 
Following the top of diagram  Eq.~\eqref{eqn:dgm:parallel_extend_basis}, we know that $e$ has a unique endpoint vertex  $v \in F(S_{\sigma_\ell})$, and that vertex has an image under $\phi_{S_{\tau_i}^n}$,  which is a connected component represented by  $\phitt(v) = w \in V_G$. 
Following down, the edge $e$ has an edge image $\phitt(e) = e' \in E_G$. 
So the question becomes: are $e'$ and $w$ in the same connected component of the slice $(V_G,E_G)_{\sigma_\ell,n+k}$, whose components represent the elements of $G(S_{\sigma_\ell}^{n+k})$? 
In order to answer this question easily for all starting edges in $F(S_{\tau_i})$, we use a breadth first search  algorithm to label all the connected components of the slice once. Then the two images of each edge starting from $F(S_{\tau_i})$ can be checked in constant time~\cite[Sec.~5.6]{Jeffe-book}. 
This results in a total time (when $d=1$) of $O(|V_G|+|E_G|)$ time taken for checking the parallelogram.
There are $2L$ dimension 1 cells in $K$, and after computing connected components for the two parallelogram diagrams, all edges in a cell can have each parallelogram diagrams checked in $O(1)$ time, thus the time to determine if 
$\max_{\sigma \leq \tau} \Lpl^{S_\tau, S_\sigma} \leq k$ 
 is $O(L \cdot (|V_G|+|E_G|))$.

In the example of Fig.~\ref{fig:Assignment}, assume $n=k=1$ and assume the given input $\phitt$ is as noted. 
Then  for the diagram of Eq.~\eqref{eqn:dgm:parallel_extend_basis} with $\ell = j$ and chasing $bc \in F(S_{\tau_j})$,
this comes down to checking if the connected component of $\phitt(b) = w$ and $\phitt(bc) = xy$ are the same in the portion of $(V_G,E_G)_{\sigma_j,2}$.
In this particular example, there are two connected components in this slice and the images are not in the same component. 
Then we know that $\Lpl^{S_{\tau_j}, S_{\sigma_j}}>k$ so we would skip all other commutative diagram checks and immediately move on in our binary search. 
If it were the case that the two images were in the same connected component, then $\Lpl^{S_{\tau_j}, S_{\sigma_j}}\leq k$ and thus we would move on to the next commutative diagram check.

Checking if $\Ltd^{S_\tau} \leq k$ is similar so we briefly highlight the differences.  
First, there are two types of basis elements in our case where $d=1$, so we need to check  $\Ltd^{S_{\sigma_i}} \leq k$ (meaning checking vertices) and $\Ltd^{S_{\tau_i}} \leq k$ (meaning checking edges). 
We focus on the case of vertices since the edge version is similar. 
For any vertex $v \in F(S_{\sigma_i})$, we need to chase it around the diagram
\begin{equation}
\label{eqn:dgm:tri_extend_basis}
    \begin{tikzcd}
        F(S_{\sigma_i}) 
            \ar[rr, "{F[S_{\sigma_i} \subseteq S_{\sigma_i}^{2n}]}"]   
            \ar[dr, "\phi_{S_{\sigma_i}}"',violet]
            & & F(S_{\sigma_i}^{2n})  \ar[r] 
            & F(S_{\sigma_i}^{2(n+k)}).
        \\
        & G(S_{\sigma_i}^n)
            \ar[ur, "\psi_{S_{\sigma_i}^n}"', orange]  
        & \substack{v\\ \phantom{x}} 
            \ar[rr, mapsto, shift left] \ar[dr, mapsto]
        & & 
        \substack{{[v]}\\{[v']}}
            \ar[r, mapsto, shift left]
            \ar[r, mapsto, shift right]
        & \substack{{[v]}\\{[v']}}\\
        & & & w \ar[ur, mapsto]
    \end{tikzcd}
\end{equation}
If $\texttt{phi}: v \mapsto w$, and $\texttt{psi}: w \mapsto v'$, the question again becomes:  are $v$ and $v'$ in the same connected component of $(V_G,E_G)_{\sigma_j,2(n+k)}$?
Similar to the parallelogram case, we take the relevant slice of the graph and check this connectivity question by finding connected components once in the slice in  $O(|V_G| +|E_G|)$ time, and then the check for each vertex in $F(S_{\sigma_i})$ is done in $O(1)$ time. 
As before, either the elements checked are in the same connected component of the relevant slice of the graph, in which case we move to the next diagram; or it does not, and we move to a different $k$ in our binary search.
There are $O(L)$ cells (counting both 0- and 1-dimensional cells) in $K$, meaning there are $O(L)$ triangle diagrams to check.
Thus, checking if $\max_{\sigma  \in K} \Ltd^{S_\sigma} \leq k $ can also be done in $O(L \cdot (|V_G| + |E_G| )$ time.

In our example case of Fig.~\ref{fig:Assignment} with $n=k=1$,  we have $2(n+k) = 4$.
Then chasing $b$, we need to check that $b$ and $\psitt \circ \phitt(b) = c$ are in the same connected component of $(V_G,E_G)_{\sigma_j,4}$. 
As this slice has one connected component, this triangle commutes. 
We can check another triangle $\Ltu^{S_{\sigma_j}} \leq k$ by chasing $w$. 
In this case, we must check if  $w$ and $\phitt \circ \psitt(w) = z$ are in the same component of $(V_G,E_G)_{\sigma_i,4}$, which again, they both are. 
In either case, if they were not, we would know the loss function is at least $k$ and continue in the binary search.

Putting this together, this means that if the graph representations of $F$ and $G$ are $(V_F,E_F)$ and $(V_G,E_G)$ respectively, the time for computing the loss function is 
\begin{equation*}
    O\Bigg( 
    L\log L  \cdot \max\Big\{|V_F|+|E_F|,|V_G|+|E_G| \Big\} \Bigg)
\end{equation*}
where the $\log L$ term comes from the binary search.

\subsection{Generalization for $d> 1$}

The generalization to higher dimensions makes relatively minor modifications for the algorithm, with the expected curse of dimensionality result in the running time. 
In this case, we build a graph with vertex set $V_F = \coprod_{\sigma \in K} F(S_\sigma)$ so that vertices represent all dimensional cubes rather than only 0 as earlier. 
An edge is given between every pair of vertices $v$ and $w$ for which $F[S_\tau \subseteq S_\sigma](v) = w$. 
Denote the sizes of these sets by $|V_F|$ and $|E_F|$. 
Note that these sizes are in some sense already hiding an exponential term in $d$ since the number of cells in the grid $K$ is $O(L^d)$.

As in $d=1$, we proceed using a binary search on $[0,\cdots,2L]$ where $[-L\delta,L\delta]^d $ is the bounding box of the images of $f: \X \to \R^d$ and $g:\Y \to \R^d$.
Again, for a fixed $k$, we will determine if $L_B(\phi,\psi) \leq k$ by checking if $\Lpl^{S_\tau, S_\sigma}$,  $\Ltd^{S_\sigma}$, 
$\Lpr^{S_\tau, S_\sigma}$ and $\Ltu^{S_\sigma}$ are all less than $k$ for all $\sigma$ and $\tau$ in $K$. 

If we count in terms of the open sets $S_\sigma$, every set  $F(S_\sigma)$ needs to be checked as the starting point for one triangle diagram $\triangled_{\Phi, \Psi}(S_\sigma)$, and as the starting point for one parallelogram diagram $\Parallelograml_{\Phi}(S_\sigma,S_\tau)$ for every $\tau \geq \sigma$. 
The grid structure means there are worst case $O(2^d)$ adjacent cells, so this results in $1+O(2^d)$ diagram checks to be done per vertex. 
Each of these checks involves determining if two vertices are in the same connected component of the higher dimensional analogue of a slice for $\sigma_{\overrightarrow{\ell}}$ with indices $\overrightarrow\ell \in \Z^d$ involves checking  
a portion of the graph with indices in a $d$-dimensional box 
$[\ell_1-(n+k), \ell_1+(n+k) ] \times \cdots \times [\ell_d-(n+k), \ell_d+(n+k) ]$
 and hence takes worst case $O(|V_F|+|E_F|)$  time.
 However as before, this connected component needs to only be found once per diagram.
 The result is a running time of 
 \begin{equation*}
O(\log L \cdot  2^d \cdot 
\max\{|V_F| + |E_F|, |V_G| + |E_G| \}).
 \end{equation*}

\section{Extension to Reeb Graphs}
\label{sec:ReebLoss}

We now take a brief diversion into understanding how the loss function framework can be used to approximate the Reeb graph interleaving distance. 
In this case, we consider a 1-dimensional mapper graph to be an approximation of the Reeb graph 
\cite{Carriere2017,Carriere2018,Munch2016,Brown2020,botnan2020}. 
We show that in order to bound the Reeb graph interleaving distance, we can compute the mapper graph for a resolution $\delta$, and then use the loss function to provide a similar bound.

\subsection{Definitions}

Given input data $f:\X \to \R$, the Reeb graph of $(\X,f)$ is computed as follows. 
Define an equivalence relation by setting $x \sim y$ iff $x$ and $y$ are in the same path-connected component of the levelset $f\inv(a)$.
With enough restrictions on the space and function (for example, a Morse function on a manifold), the resulting Reeb graph is a topological graph; i.e.~a 1-dimensional stratified space. 
Similar to the vantage taken for the mapper graphs in this paper, the data of a Reeb graph can be stored in a cosheaf. 
\begin{definition}
    For a given $(\X,f)$, the associated Reeb cosheaf is given by 
    \begin{equation*}
        \begin{matrix}
            \tF: & \Int & \to & \Set \\
            & I & \mapsto & \pi_0 f\inv(I)\\
            & \rotatebox[origin=c]{-90}{$\subseteq$} &  & \downarrow \pi_0[\subseteq]\\ 
            & J & \mapsto & \pi_0 f\inv(J)\\
        \end{matrix}
    \end{equation*}
    where morphisms are induced by the $\pi_0$ functor. 
\end{definition}

For clarity, we write the Reeb cosheaf with a tilde to distinguish it from the mapper cosheaf without a tilde. 
Given this input, we have the Reeb graph interleaving distance \cite{deSilva2016}, given as follows.
\begin{definition}
\label{def:ReebInterleavingDistance}
    Define the functor $(-)^\e: \Int \to \Int$ by $(a,b) \mapsto (a-\e,b+\e)$ with morphisms induced by inclusion. 
    Then $\tF_\e:\Int \to \Set$ is given by $\tF_\e(J) = \tF(J^\e)$.

    For given $\tF,\tG: \Int \to \Set$, an $\e$-interleaving is a pair of natural transformations $\tphi:\tF \Rightarrow \tG_\e$ and $\tpsi:\tG \Rightarrow \tF_\e$ such that 
\begin{equation*}
    \begin{tikzcd}
        \tF(I) 
            \ar[rr, "{\tF[I \subseteq I^{2n}]}"]   
            \ar[dr, "\tphi_I"',violet]
            & & \tF(I^{2n}) & 
        & \tF(I^n) \ar[dr]
            \ar[dr, "\tphi_{I^{n}}",violet]
        & \\
        & \tG(I^n)\ar[ur, "\tpsi_{I^n}"', orange]  & & 
        \tG(I) 
            \ar[rr, "{\tG[I \subseteq I^{2n}]}"']
            \ar[ur, "\tpsi_{I}", orange] 
        && \tG(I^{2n})
    \end{tikzcd}
\end{equation*}
    commute for all $I \in \Int$. 
    The (categorical) Reeb graph interleaving distance is given by 
    \begin{equation*}
        d_R(\tF,\tG) = \inf\{ \e \geq 0 \mid \text{ there exists an $\e$-interleaving}\}.
    \end{equation*}
    
\end{definition}

Fix a $\delta$.
Following Sec.~\ref{sec:background}, denote the vertices of $K$ by $\{\sigma_{-L},\cdots,\sigma_L\}$ where $\sigma_i$ is at the point $i\delta \in \R$. 
Denote the edges by $\tau_i = (i\delta,(i+1)\delta)$ which has faces $\sigma_i$ and $\sigma_j$. 
Given some input data $f:\X \to \R$, we can either construct its Reeb cosheaf $\tF:\Int \to \Set$, or by fixing some choice of $\delta$, we can construct its mapper cosheaf $F:\Open(K) \to \Set,\, F(S) = f\inv(|S|)$. 

We next show that the loss function we have computed here on the mapper version $d_I$ can be used to similarly bound the Reeb interleaving distance $d_R$.
We do this by showing that $d_I$ is an approximation of $d_R$, which can be viewed as a special case of \cite[Thm.~5.15]{botnan2020};  however, for clarity, we include a direct proof in Sec.~\ref{sec:technicalProofs} as our setting allows a proof with considerably less use of category theoretic machinery.

\begin{restatable}{proposition}{reebvsmapperbound}
\label{prop:reebvsmapperbound}
    For inputs $f:\X \to \R$ and $g:\Y \to \R$, denote the respective Reeb cosheaves as $\tF,\tG:\Int \to \Set$, and the respective mapper cosheaves as $F,G: \Open(K) \to \Set$. Then 
    \begin{equation*}
        d_R(\tF,\tG) \leq \left(d_I(F,G) + 1\right)\delta. 
    \end{equation*}
\end{restatable}

Given this bound, we combine Prop.~\ref{prop:reebvsmapperbound} with Thm.~\ref{thm:secondBound} to show that the loss function for the mapper graph discretization bounds the Reeb graph interleaving as well and that, in particular, this bound is controlled by the diameter $\delta$ chosen for $K$. 

\begin{cor}
Given a basis $n$-assignment  
$\phi = \{\phi_{U_\sigma} \mid \sigma \in K\}$ 
and 
$\psi = \{\psi_{U_\sigma} \mid \sigma \in K\}$ for $F,G: \Open(K) \to \Set$, we have that 
\begin{equation*}
    d_R(\tF,\tG) \leq \delta(d_I(F,G) + 1)  \leq \delta(n + L_B(\phi,\psi)+1).
\end{equation*}
\end{cor}

\section{Technical Proofs}
\label{sec:technicalProofs}

In this section, we include the technical proofs from the previous sections. 

\subsection{Proofs from Sec.~\ref{sec:background}}

\begin{lemma}
\label{lem:thickeningIsFunctor}
$(-)^n$ is a functor. 
\end{lemma}

\begin{proof}
First, we check that the images of morphisms are well defined, which is to say that if $S \subseteq T$, then $S^{n} \subseteq T^{n}$. 
The statement is clear if $n = 0$, so by induction, we assume that $S^{n-1} \subseteq T^{n-1}$. 
Given an arbitrary $U_\sigma \in S^{n}$, the statement is immediate if $U_\sigma \in S^{n-1} \subseteq S^{n}$, so we assume $U_\sigma \in S^n \setminus S^{n-1}$. 
For this to happen, there must be a $U_\gamma \in S^{n-1}$ and $\tau \in K$ with $\gamma \geq \tau \leq \sigma$ and thus $U_\gamma \subseteq U_\tau \supseteq U_\sigma$. 
But as $U_\gamma \in S^{n-1} \subseteq T^{n-1}$, this sequence also implies that $U_\sigma \in T^n$,  finishing the well-defined check. 

To ensure this is a functor, we need to check that the identity morphism is sent to the identity, and that composition holds. 
For the former, we see that $S \subseteq S$ gets sent to $S^n \subseteq S^n$, and each is an identity. 
The latter is immediate from the property that $\Open(\cU)$ is a poset category, meaning that there is at most one morphism between any pair of objects.  
\end{proof}

One property of this construction that will be useful is as follows. 
For any $U_\sigma \in S^n$, there is a $U_\tau \in S$ and a sequence of cells of $K$
\begin{equation}
\label{eq:length_n_path_cells}
    \tau 
    \geq \gamma_1 \leq \tau_1 
    \geq \gamma_2 \leq \tau_2 
    \geq  \cdots 
    \geq \gamma_n \leq \sigma
\end{equation}
and thus also a sequence of sets in $\cU$
\begin{equation}
\label{eq:length_n_path_sets}
    U_\tau 
    \geq U_{\gamma_1} \leq U_{\tau_1} 
    \geq U_{\gamma_2} \leq U_{\tau_2} 
    \geq  \cdots 
    \geq U_{\gamma_n} \leq U_\sigma.
\end{equation}
Further, given such a sequence with $\tau \in U$, we know that $\sigma \in U^n$. 
Two examples of this can be seen in Fig.~\ref{fig:Length_n_path}, where $\sigma$ and $\sigma'$ from $U^3$ are given, along with a path satisfying Eq.~\eqref{eq:length_n_path_cells}.
Of course, the choice of sequence for Eq.~\eqref{eq:length_n_path_cells} is not unique, so other options are possible. 
\begin{figure}
    \centering
    \includegraphics[width = 0.35\textwidth]{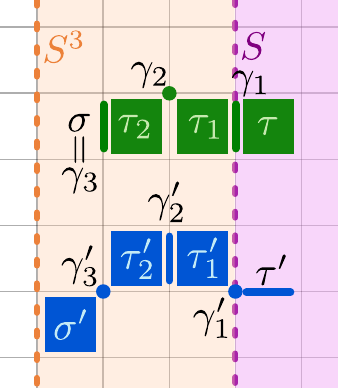}
    \caption{
    Given the purple set $U$, we have an edge $\sigma$ and a square $\sigma'$ which are elements of $U^3$. 
    Then we provide an example sequence for each satisfying Eq.~\eqref{eq:length_n_path_cells} leading to $\tau$ and $\tau'$ in $U$.
    }
    \label{fig:Length_n_path}
\end{figure}

Next we show that the distance of Eq.~\eqref{def:interleavingDistance} is indeed a distance 
using the super-linear family of translations framework of \cite{Bubenik2014a}.
This construction can be generalized to the concept of a category with a flow \cite{deSilva2018}, but the added generality is not needed here. 

\begin{definition}[\cite{Bubenik2014a}]
\label{defn:superlinear}
Let $P = (P,\leq)$ be a preordered set. 
A \emph{translation} on $P$ is a functor $\Gamma: P \to P$ along with a natural transformation $\eta:\1_P \Rightarrow \Gamma$. 
A \emph{super-linear family of translations} is a collection $\{\Gamma_\e \}_{\e\geq 0}$  such that 
$\Gamma_\e \Gamma_{\e'}(p) \leq \Gamma_{\e + \e'}(p)$ for all $p \in P$, and $\e, \e' \geq 0$. 
\end{definition}

\composedThickenings*

\begin{proof}
First, we check that $(-)^n$ is indeed a translation using the above terminology. 
In particular, we define $\gamma^n:\1_{\Open(\cU)} \Rightarrow (-)^n$ to have components $\gamma^n_S: S \to S^n$ as simply the inclusion, and we can easily check that this satisfies the naturality requirements. 

Fix $S \in \Open(\cU)$. 
We need to show that $(S^n)^{n'} = S^{n+n'}$. 
Let $U_\sigma \in (S^n)^{n'}$. 
By previous remarks, this is true if and only if there is a sequence in $K$
\begin{equation*}
    \tau 
    \geq \gamma_1 \leq \tau_1 
    \geq \gamma_2 \leq \tau_2 
    \geq  \cdots 
    \geq \gamma_{n'} \leq \sigma
\end{equation*}
with $U_\tau \in S^n$. 
But this property of $\tau$ happens iff there is also a sequence in $K$
\begin{equation*}
    \tau' 
    \geq \gamma_1' \leq \tau_1 '
    \geq \gamma_2' \leq \tau_2' 
    \geq  \cdots 
    \geq \gamma_{n}' \leq \tau
\end{equation*}
with $U_{\tau'} \in S$. 
Concatenating the two sequences gives a sequence of length $(n+n')$ from $U_\tau \in S$ to $U_\sigma$. 
Thus $U_\sigma \in S^{n+n'}$ iff $U_\sigma \in (S^n)^{n'}$, and hence $(S^n)^{n'} = S^{n+n'}$. 
\end{proof}

\begin{theorem}
    The interleaving distance of Defn.~\ref{def:interleavingDistance} is an extended pseudometric. 
\end{theorem}
\begin{proof}
Because Lem.~\eqref{lem:composedthickenings} is a stronger requirement than needed for Defn.~\ref{defn:superlinear}, the collection $\{ ( - )^n\}_{n \geq 0}$ forms a super-linear family of translations. 
Then the result is immediate from \cite[Theorem 3.21]{Bubenik2014a}.
\end{proof}

\subsection{Proofs from Sec.~\ref{sec:loss-function}}

To simplify notation, throughout the proofs of this section we often use a $\bullet$ symbol to represent the set indexing a particular map when the subscript would be obvious from the given map. 
For example, we write $\phi_\bullet: F(S^n) \to G(S^{2n})$ rather than writing $\phi_{S^n}$.

\lossimpliescommutes*

\begin{proof}[Proof of Lem.~\ref{lem:lossimpliescommutes}]
We prove the lemma for the first and third entries only as the other arguments are symmetric. 
Assume $\Lpl^{S,T}(\phi) \leq k$ and 
consider the diagram
\begin{equation}
\label{eq:dgm_cheese_wedge}
\begin{tikzcd}[column sep = 4em]
F(S) 
\ar[r, "{F[\subseteq]}"] 
\ar[dr, "\phi_S",] \ar[ddr, "\Phi_S"', very near start]
& F(T) \ar[dr, "\phi_T", very near start] \ar[ddr, "\Phi_T"',very near start]\\
&    G(S^n) \ar[r, crossing over,  very near start, "{G[\subseteq]}"'] \ar[d, "{G[\subseteq]}"] & G(T^n) \ar[d, "{G[\subseteq]}"] \\
&    G(S^{n+k}) \ar[r,  "{G[\subseteq]}"] & G(T^{n+k}).
\end{tikzcd}
\end{equation}
Note that the top of the diagram (Eq.~\ref{eq:dgm_cheese_wedge}) given by
\begin{equation}
\label{eq:dgm_cheese_top}
\begin{tikzcd}
    F(S) \ar[r, "{F[\subseteq]}"] \ar[d,"\phi_S"]
        & F(T) \ar[d, "\phi_T"]\\
    G(S^n) \ar[r, "{G[\subseteq]}"]  
        & G(T^n)
\end{tikzcd}
\end{equation}
does not necessarily commute in the case that $k\geq 1$, and the bottom of the diagram (Eq.~\ref{eq:dgm_cheese_wedge}) given by 
\begin{equation}
\label{eq:dgm_cheese_bottom}
\begin{tikzcd}
    F(S) \ar[r, "{F[\subseteq]}"] \ar[d,"\Phi_S"]
        & F(T) \ar[d, "\Phi_T"]\\
    G(S^{n+k}) \ar[r, "{G[\subseteq]}"]  
        & G(T^{n+k})
\end{tikzcd}
\end{equation}
is $\Parallelograml_{\Phi}(S,T)$, for which we wish to check for commutativity. 
For any $x \in F(S)$, following around the top square, Eq.~\ref{eq:dgm_cheese_top}, gives 
\begin{equation*}
\begin{tikzcd}
x \ar[r,mapsto] \ar[d,mapsto]
& x' \ar[d, mapsto] \\
\substack{ \\a} 
    \ar[r, mapsto, shift right,end anchor = {[yshift = -0.4ex]}] 
& \substack{ b'\\a'}
\end{tikzcd}
\end{equation*}
with $d_{T^n}^G(a',b') \leq k$.
By definition, the image of $a'$ and $b'$ is the same under the map \linebreak ${G(T^n) \to G(T^{n+k})}$.
Then since the front square of Eq.~\ref{eq:dgm_cheese_wedge} given by 
\begin{equation*}
\begin{tikzcd}
    G(S^n) \ar[r, "{G[\subseteq]}"] \ar[d,"{G[\subseteq]}"]
        & G(T) \ar[d, "{G[\subseteq]}"]\\
    G(S^n) \ar[r, "{G[\subseteq]}"]  
        & G(T^n)
\end{tikzcd}
\end{equation*}
commutes by functoriality of $G$, and the side triangles of Eq.~\ref{eq:dgm_cheese_wedge} given by 
\begin{equation*}
\begin{tikzcd}
    F(S) \ar[r, "\phi_S"] \ar[dr, "\Phi_S"']  & G(S^n) \ar[d, "{G[\subseteq]}"] 
    & F(T) \ar[r, "\phi_T"] \ar[dr, "\Phi_T"'] & G(T^n) \ar[d, "{G[\subseteq]}"] \\
    & G(S^{n+k}) && G(T^{n+k})
\end{tikzcd}
\end{equation*}
commute by definition of $\Phi$, we have that the image of $x$ under either direction of the back square, Eq.~\ref{eq:dgm_cheese_bottom},
commutes, proving claim (1). 

Turning to claim (3), consider the noncommutative diagram 
\begin{equation}
\label{eq:dgm_colors}
\begin{tikzcd}[execute at end picture={
\foreach \Nombre in  {A,B,...,F}
  {\coordinate (\Nombre) at (\Nombre.center);}
\fill[yellow,opacity=0.3] 
  (A) -- (E) -- (F) -- cycle;
\fill[yellow,opacity=0.3] 
  (F) -- (C) -- (D) -- cycle;
\fill[blue,opacity=0.1] 
  (B) -- (C) -- (F) -- (E) -- cycle;
}]
|[alias=A]| F(S) 
    \ar[rrr, "{F[\subseteq]}"] 
    \ar[dr, "\phi_\bullet"']
    \ar[drr, "\Phi_\bullet"]
&&& |[alias=B]| F(S^{2n}) 
    \ar[r, "{F[\subseteq]}"] 
& |[alias=C]| F(S^{2n+k}) 
    \ar[r, "{F[\subseteq]}"] 
& |[alias=D]| F(S^{2(n+k)}).
\\
& |[alias=E]| G(S^n) 
    \ar[r, "{G[\subseteq]}"'] 
    \ar[urr, "\psi_\bullet"]
& |[alias=F]| G(S^{n+k}) 
    \ar[urr, "\psi_\bullet"]
    \ar[urrr, "\Psi_\bullet"']
\end{tikzcd}
\end{equation}
The  two yellow triangles 
\begin{equation*}
\begin{tikzcd}
    F(S) \ar[dr, "\phi_{\bullet}"'] \ar[drr, "\Phi_{\bullet}"]
        &&& F(S^{2n+k}) \ar[r, "{F[\subseteq]}"]  & F(S^2(n+k)\\
    & G(S^n)\ar[r, "{G[\subseteq]}"']  & G(S^{n+k}) \ar[ur, "\psi_{\bullet}"] \ar[urr, "\Psi_{\bullet}"] 
\end{tikzcd}
\end{equation*}
commute by definition of $\Phi$ and $\Psi$. 
The blue parallelogram 
\begin{equation*}
\begin{tikzcd}
    G(S^{n}) \ar[r, "{G[\subseteq]}"] \ar[d, "\psi_\bullet"]  & G(S^{n+k}) \ar[d, "\psi_\bullet"]\\
    F(S^{2n}) \ar[r, "{G[\subseteq]}"] & F(S^{2n+k})
\end{tikzcd}
\end{equation*}
is the diagram $\Parallelogramr_\psi(S^n,S^{n+k})$ which also has loss function bounded by $k$, thus elements of $G(S^n)$ are not necessarily the same in the image of $F(S^{2n+k})$ following the parallelogram, but are the same in $F(S^{2(n+k)})$. 

Checking that $\triangled_{\Phi,\Psi}(S)$ commutes amounts to a diagram chase.
For an arbitrary $\alpha \in F(S)$, consider the following elements
\begin{equation*}
\begin{tikzcd}
\alpha
    \ar[rrr, mapsto, end anchor = {[yshift = 1ex]}] 
    \ar[dr, mapsto]
    \ar[drr, mapsto]
&&& \substack{a\\a'\\}
    \ar[r, mapsto, 
        end anchor = {[yshift = 1.5ex]},
        start anchor = {[yshift = 1ex]}] 
    \ar[r, mapsto, 
        end anchor = {[yshift = -0.2ex]},
        start anchor = {[yshift = -1ex]}] 
& \substack{b\\b'\\b''} \Big\}
    \ar[r, mapsto] 
& c
\\
& x
    \ar[r, mapsto] 
    \ar[urr, mapsto,end anchor = {[yshift = -.5ex]}]
& x'
    \ar[urr,mapsto, 
        end anchor = {[yshift = -1ex]}]
    \ar[urrr, mapsto, bend right = 10, start anchor = {[yshift = -.5ex]}]
\end{tikzcd}
\end{equation*}
 aligning with the diagram of Eq.~\ref{eq:dgm_colors}. 
Both $\alpha$ and $x$ map to $x'$ because of the yellow triangle commuting, and both $b''$ and $x'$ map to the same $c$ for the same reason.
Even if $\alpha$ and $x$ map to different elements in $F(S^{2n})$, they must map to the same element in $F(S^{2(n+k)})$, and this element must be $c$, since both $b'$ and $b''$ map to the same element by the bound on the blue parallelogram. 
As this was done for an arbitrary $\alpha$, we have that $\triangled_{\Phi, \Psi}(S)$ commutes. 

Claims (2) and (4) are similar with appropriate choices of diagrams. 
The final statement is immediate since $L(\phi,\psi) = 0$ implies all diagrams needed for an interleaving commute.
\end{proof}

\extendToNatTrans*
\begin{proof}[Proof of Lem.~\ref{lem:extendToNatTrans}] 
We start by defining $\Phi_S$ for arbitrary open sets. 
Note that since ${\Lpl^{S_\tau, S_\sigma} = 0}$, for any $\sigma \leq \tau$, the diagram of the form 
\begin{equation*}
\begin{tikzcd}
F(S_\tau) \ar[r,"{F[\subseteq]}"] \ar[d] \ar[d, "{\Phi_{S_\tau}}"']
    & F(S_\sigma) \ar[d, "{\Phi_{S_\sigma}}"] \\
G(S_\tau^n) \ar[r,"{G[\subseteq]}"] 
    & G(S_\sigma^n)
\end{tikzcd}
\end{equation*}
commutes.

For an arbitrary open $S$, define 
$\cU_S = \{ S_\sigma \mid U_\sigma \in S\}$. 
Any nonempty intersection $U_\sigma \cap U_\tau$ is also an element of $\cU$, and so any nonempty intersection of $S_\sigma \cap S_\tau$ is an element of $\cU_S$, so it is a cover of $S$. 
Then we use the fact that $F$ is a cosheaf, and in particular this means that $F(S)$ is the coequalizer of the diagram 
\begin{equation*}
\begin{tikzcd}[column sep = 1in]
    \displaystyle
    \coprod_{\sigma, \sigma'} F(S_\sigma \cap S_{\sigma'}) 
        \ar[r, shift left, "{F[S_\sigma \cap S_{\sigma'}\subseteq S_\sigma]}"] 
        \ar[r, shift right, "{F[S_\sigma \cap S_{\sigma'}\subseteq S_{\sigma'}]}"'] 
    &
    \displaystyle\coprod_{\tau} F(S_\tau).
\end{tikzcd}
\end{equation*} 
Rephrased, this means that for any set $Q$ with maps $F(S_\sigma) \to Q$ such that the  solid arrow diagrams of the form 
\begin{equation*}
\begin{tikzcd}
F(S_\sigma \cap S_{\sigma'}) \ar[d, "{F[\subseteq]}"'] \ar[r,"{F[\subseteq]}"] 
& F(S_\sigma) \ar[d, "{F[\subseteq]}"]  \ar[ddr, bend left] \\
F(S_{\sigma'}) \ar[r, "{F[\subseteq]}"] \ar[drr, bend right] 
& F(S)  \ar[dr, dashed, "\exists!"]\\
 & & Q
\end{tikzcd}
\end{equation*}
commute for any $\sigma, \sigma'$, then there is a unique map $F(S) \to Q$ whose addition still has all diagrams commute. 
In our case, set $Q = G(S^n)$, and define the legs of the cocone to be $G[\subseteq] \circ \Phi_{S_\sigma}$ as seen in the bold purple arrows of the diagram 
\begin{equation}
\label{eqn:cubedgm}
\begin{tikzcd}[row sep=1.5em, column sep = 1.5em]
F(S_{\sigma} \cap S_{\sigma'})
    \arrow[rr, "\Phi_{\bullet}"] \arrow[dr, "{F[\subseteq]}"] 
    \arrow[dd,swap, "{F[\subseteq]}"] 
    &&
G( (S_{\sigma} \cap S_{\sigma'})^n)\arrow[dd, "{G[\subseteq]}"', very near start] \arrow[dr, "{G[\subseteq]}"] \\
& 
F(S_{\sigma'})
    \arrow[rr, crossing over, thick, violet,"\Phi_{\bullet}", near start ] 
&&
G(S_{\sigma'}^n)
    \arrow[dd, thick, violet, "{G[\subseteq]}"] \\
F(S_\sigma) 
    \arrow[rr, thick, violet,"\Phi_{\bullet}", near start] 
    \arrow[dr, "{F[\subseteq]}"] 
&& 
G(S_{\sigma'}^n)\arrow[dr, thick,  violet, "{G[\subseteq]}"] \\
& 
F(S)
    \arrow[rr, dashed, "{\exists! \, \, \Phi_S}"] 
    \arrow[uu, leftarrow, crossing over, "{F[\subseteq]}", very near end]
&& G(S^n).
\end{tikzcd}
\end{equation}
Note that the diagram prior to the inclusion of the dotted line commutes, since we can check the relevant faces as follows. 
The left and right squares commute because $F$ and $G$ are functors. 
The back and top panels commute because they involve only basis opens; equivalently, because we assumed $\Lpl^{S_\sigma \cap S_{\sigma'}, S_\sigma}= \Lpl^{S_\sigma \cap S_{\sigma'}, S_{\sigma'}}  = 0$.
Then, because $F(S)$ is a colimit of the diagram, there exists a unique map $\Phi_S:F(S) \to G(S^n)$ as noted, making any diagram of this form commute. 

To ensure that the resulting $\Phi_S$ maps  make diagrams of the form 
\begin{equation*}
\begin{tikzcd}
    F(S) \ar[r] \ar[d] & G(S^n) \ar[d] \\
    F(T) \ar[r] & G(T^n)
\end{tikzcd}
\end{equation*}
commute for arbitrary $S \subseteq T$, fix such a pair and an $x \in F(S)$. 
Because $F(S)$ is the colimit, there is a $\sigma$ and an $x_\sigma \in F(S_\sigma)$ such that $x_\sigma \mapsto x$. 
In this case we have the diagram 
\begin{equation*}
\begin{tikzcd}
F(S_\sigma) 
\ar[r, "\Phi_\bullet"] 
\ar[dr] \ar[ddr]
& G(S_\sigma)^n \ar[dr] \ar[ddr]\\
&    F(S) \ar[r, crossing over, "\Phi_\bullet", very near start] \ar[d] & G(S^n) \ar[d] \\
&    F(T) \ar[r, "\Phi_\bullet"] & G(T^n)
\end{tikzcd}
\end{equation*}
The top and bottom squares of the wedge commute because they are the front of the cube of the diagram in Eq.~\eqref{eqn:cubedgm}. 
The left and right triangles commute since $F$ and $G$ are functors. 
Thus the front square commutes. 
This means the resulting $\Phi$ is a natural transformation, and thus $\Lpl^{S,T} = 0$. 
\end{proof}

\extendTriangles*
\begin{proof}[Proof of Lem.~\ref{lem:extendTriangles}] 
Because $\Ltd^{S_\sigma} = 0$ for all basis elements, diagrams of the form 
\begin{equation*}
\begin{tikzcd}
    F(S_\sigma) \ar[rr, "{F[\subseteq]}"] \ar[dr, "\Phi_{S_\sigma}"']&& F(S_\sigma^{2n})\\
    & G(S_\sigma^n) \ar[ur, "\Psi_{S_{\sigma}^n}"']
\end{tikzcd}
\end{equation*}
commute for any $\sigma \in K$.
Given an arbitrary open set $S$, let $x \in F(S)$ be given. 
As in the proof of Lem.~\ref{lem:extendToNatTrans}, there is a $\sigma$ and an $x_\sigma \in F(S_\sigma)$ with $x_\sigma \mapsto x$. 
Then consider the diagram 
\begin{equation*}
\begin{tikzcd}
	F(S_\sigma) &&&& F(S_\sigma^{2n}) \\
	&& F(S) &&& {} & F(S^{2n}) \\
	&& G(S_\sigma^n) \\
	&&&& G(S^n).
	\arrow[from=4-5, to=2-7, "\Psi_\bullet"']
	\arrow[from=3-3, to=1-5, "\Psi_\bullet", near end]
	\arrow[from=1-1, to=1-5, "{F[\subseteq]}"]
	\arrow[from=1-1, to=3-3, "\Phi_\bullet"']
	\arrow[from=1-1, to=2-3, "{F[\subseteq]}"]
	\arrow[from=1-5, to=2-7, "{F[\subseteq]}"]
	\arrow[from=3-3, to=4-5, "{G[\subseteq]}"']
	\arrow[from=2-3, to=2-7, crossing over, "{F[\subseteq]}"]
	\arrow[from=2-3, to=4-5, crossing over, "\Phi_\bullet", near end]
\end{tikzcd}
\end{equation*}
The top square commutes because $F$ is a functor. 
The back triangle commutes by this lemma's assumption.
The left and right squares commute because $\Phi$ and $\Psi$ are natural transformations.
Taken together, this means that the front triangle commutes as required.
\end{proof}

\subsection{Proof from Sec.~\ref{sec:ReebLoss}}
\reebvsmapperbound*
\begin{proof}
    Let $\phi,\psi$ be an $n$-interleaving for $F,G:\Open(\cU) \to \Set$. 
    We will construct an $\e=\delta(n+1)$-interleaving $\tphi$, $\tpsi$ for $\tF,\tG:\Int \to \Set$. 

    We start by defining $\tphi : \tF \Rightarrow \tG^\e$ as $\tpsi$ is analogous. 
    Given an arbitrary interval $I= (a,b)$, let $J = (j\delta, k\delta)$ be the smallest grid-aligned interval containing $I$; i.e. ${j\delta \leq a < (j+1)\delta}$ and $(k-1)\delta<b \leq k\delta$.
 
    Note that 
    $I \subseteq J \subseteq J^{\delta n} \subseteq I^{(n+1)\delta} = I^\e$.
    Let $S = \{S_{\tau_i} \mid j \leq i \leq k-1 \} \cup \{ S_{\sigma_i} \mid j < i < k \}$. 
    A quick check shows that ${S \in \Open(\cU)}$, that $J = |S|$, and that $J^{\delta n} = |S^n|$.
    Chasing definitions, this means that ${\tF(J) = \pi_0(f\inv(J))}$ and $F(S) = \pi_0(f\inv(|S|))$ are equal; similarly $\tF(J^{\delta n}) = F(S^n)$. 
    Then define $\tphi_I$ to be the map defined by the composition 
    \begin{equation*}
    \begin{tikzcd}
    \tF(I) \ar[rr,dashed, "\tphi_I"] \ar[d, "{\tF[\subseteq]}"'] 
            && \tG(I^{(n+1)\delta}) \\
        \tF(J)  \ar[d, "="'] \ar[r,dashed, "\tphi_J"]
        & \tG(J^{\delta n}) \ar[ur, "{\tG[\subseteq]}"'] \\ 
        F(S) \ar[r, "\phi_S"'] 
        & G(S^n). \ar[u, "="']
    \end{tikzcd}
    \end{equation*}
    Notice that setting $I$ to be an axis aligned interval $J$ gives the map $\tphi_J$ marked. 

    Now that we have built $\tphi$ and $\tpsi$, we need to check (i) that each is a natural transformation and (ii) that they satisfy the triangle diagrams of Defn.~\ref{def:ReebInterleavingDistance}. 
    For (i) we check only $\tphi$ as, again, $\tpsi$ is symmetric. 
    To this end, assume we have $I \subseteq I'$ with minimal grid-aligned intervals $J$ and $J'$, and let $S,S' \in \Open(\cU)$ be such that $|S|=J$ and $|S'|=J'$.
    Then consider the diagram 
    \begin{equation*}
    \begin{tikzcd}
    \tF(I) \ar[rr, "\tphi_I"] \ar[d, "{\tF[\subseteq]}"']
            && \tG(I^{(n+1)\delta}) 
            \ar[dddr, "{\tG[\subseteq]}"] \\
        \tF(J)  \ar[d, "="'] 
        && \tG(J^{\delta n}) \ar[u, "{\tG[\subseteq]}"] \\ 
        F(S) \ar[rr, "\phi_S"'] 
            \ar[dddr, "{F[\subseteq]}"']
        && G(S^n) \ar[u, "="']
            \ar[dddr, "{G[\subseteq]}"']\\
    &\tF(I') \ar[rr,"\tphi_{I'}", crossing over , near start] \ar[d, "{\tF[\subseteq]}"'] 
        \ar[uuul, leftarrow, crossing over, "{\tF[\subseteq]}"']
            && \tG((I')^{(n+1)\delta}) \\
    &    \tF(J')  \ar[d, "="] 
        && \tG((J')^{\delta n}) \ar[u, "{\tG[\subseteq]}"'] \\ 
    &    F(S') \ar[rr, "\phi_{S'}"'] 
        && G((S')^n). \ar[u, "="']    
    \end{tikzcd}
    \end{equation*}
    Note that the front and back panels of the cube are the diagrams that were used to define $\tphi_I$ and $\tphi_{I'}$, so they commute. 
    The bottom panel commutes because $\phi$ is a natural transformation. 
    The left and right panels commute because $F$ and $\tF$ arise from computing connected components on the same underling input data. 
    Thus, the top square commutes, and this is exactly what is needed to say that $\tphi$ is a natural transformation. 

    To check (ii), fix an interval $I$ with grid aligned $J \subseteq I$ and $S \in \Open(\cU)$ with $|S|=J$. 
    Then consider the diagram 
    \begin{equation*}
    \begin{tikzcd}
    \tF(I) 
        \ar[rr, "\tphi_I"] \ar[d, "{\tF[\subseteq]}"'] 
    && \tG(I^{(n+1)\delta})
        \ar[rr, "\tpsi_{I^\e}"]
        \ar[d]
    && \tF(I^{2(n+1)}\delta)
    \\
    \tF(J)  
        \ar[d, "="'] 
    & \tG(J^{\delta n}) 
        \ar[ur, "{\tG[\subseteq]}"] 
        \ar[r, "{\tG[\subseteq]}"']
    & \tG(J^{(n+1)\delta}))  
        \ar[d, "="'] 
    & \tF(J^{\delta (2n+1)}) 
        \ar[ur, "{\tF[\subseteq]}"'] 
        \\ 
    F(S) \ar[r, "\phi_S"] 
         \ar[drr, "{F[\subseteq]}"']
    & G(S^n)
        \ar[u, "="']
        \ar[r, "{G[\subseteq]}"]
        \ar[dr, "\psi_{S^n}"]
    &G(S^{n+1}) \ar[r, "\psi_{S^{n+1}}"] 
    & F(S^{2n+1}) \ar[u, "="']
    \\
    && F(S^{2n}).
    \ar[ur, "{F[\subseteq]}"']
    \end{tikzcd}
    \end{equation*}
    The left and right hexa-laterals commute by definition of $\tphi$ and $\tpsi$ respectively. 
    The middle top triangle commutes because $\tG$ is a functor, and the middle square commutes because $\tG$ and $G$ are defined as connected components of the same input data. 
    The bottom left triangle commutes because $\phi$ and $\psi$ are an $n$-interleaving. 
    The right quadrilateral commutes because $\psi$ is a natural transformation.
    All this shows that the outside boundary of the diagram commutes. 
    Swapping out the interior, we have 
    \begin{equation*}
    \begin{tikzcd}
    & \tG(I^{(n+1)\delta})
        \ar[dr, "\tpsi_{I^\e}"]\\
    \tF(I) 
        \ar[ur, "\tphi_I"] \ar[d, "{\tF[\subseteq]}"'] 
        \ar[rr, "{\tF[\subseteq]}"]
    & 
    & \tF(I^{2(n+1)}\delta)
    \\
    \tF(J)  
        \ar[d, "="'] 
        \ar[rr, "{\tF[\subseteq]}"]
    & 
    & \tF(J^{\delta (2n+1)}) 
        \ar[u, "{\tF[\subseteq]}"'] 
        \\ 
    F(S) 
         \ar[dr, "{F[\subseteq]}"']
         \ar[rr, "{F[\subseteq]}"]
    &
    & F(S^{2n+1}) \ar[u, "="']
    \\
    & F(S^{2n}).
    \ar[ur, "{F[\subseteq]}"']
    \end{tikzcd}
    \end{equation*}
    The bottom triangle commutes because $F$ is a functor, the next square up commutes by definition of $F$ and $\tF$, and the top square commutes because $\tF$ is a functor. 
    Combining this with the outside ring commuting means that the top triangle commutes, which is the final ingredient needed for the definition of an interleaving. 
\end{proof}
\section{Discussion}
\label{sec:discussion}

In this paper, we defined a loss function that quantifies how far a diagram is from being commutative, and used such a loss function to bound the interleaving distance, both for mapper and Reeb graph settings.
This work provides a way to evaluate a particular set of maps, which immediately suggests the question of utilizing this quantification to iteratively improve our comparison. 
Here, the quality of the bound is dependent on the quality of the input $n$-assignment, but we assume no control over that input in this paper and so we cannot evaluate the tightness of the bound. 
In the followup work \cite{Chambers2025}, we will use this bound in the context of an ILP framework, where an input $n$-assignment can be improved incrementally thus finding a better bound on the distance. 
The potential for not only getting better approximations but also returning the actual interleaving maps used in the bound is an exciting step toward computing interleaving distances for graph-based signatures available in practice. 
Of course, we know that deciding if two Reeb graphs are $\epsilon$-interleaved (for $\epsilon \ge 1$) is NP-hard
\cite{Bjerkevik2018}, so our ILP has no guarantee of reaching the global optimal solution. 
This is related to other available work providing bounds similarly lacking tightness guarantees for the interleaving distance restricted to merge trees \cite{Curry2022, Pegoraro2021}. 
The first of these (\cite{Curry2022}) uses the labeled merge tree distance \cite{Munch2018,Gasparovic2019} which has more limited properties when generalizing to Reeb graphs \cite{Lan2024} so it is unclear if this would generalize.
The latter uses the map formulations on the topological spaces that arise in the special case of merge trees, so major modifications would need to be made for this to apply to the Reeb graph setting as seen here. 

We believe that our loss function based framework is applicable in a broader context where data are modeled as sheaves or cosheaves in the category of sets, as sheaf theory is emerging as a tool in data science to study, e.g.,   distributed systems~\cite{Malcolm2009,Mansourbeigi2017}, sensor networks~\cite{Robinson2017}, model fit~\cite{KvingeJeffersonJoslyn2021}, and uncertainty quantification~\cite{JoslynCharlesDePerno2020}. 
In particular, one interesting next step is to study how to extend our framework to work with persistence modules as cosheaves in the category of vector spaces (e.g.,~\cite{BubenikMilicevic2021}). 
As the interleaving distance for multiparameter persistence modules is similarly NP-hard \cite{Bjerkevik2019}, this would be an exciting step toward computational efforts in this broad class of topological signatures.

\printbibliography

\end{document}